\documentclass[11pt]{article}
\usepackage[margin=1in]{geometry}
\usepackage{amsfonts,amsmath,latexsym,amssymb}
\usepackage[amsmath,amsthm,thmmarks]{ntheorem}
\usepackage{graphicx}
\usepackage{hyperref}
\usepackage[utf8]{inputenc}
\usepackage[ruled]{algorithm}
\usepackage{algorithmicx}
\usepackage{algpseudocode}
\usepackage{enumerate}
\usepackage{comment}
\usepackage[usenames]{color}
\definecolor{LinksColor}{rgb}{0,0.2,0.8}
\hypersetup{colorlinks,linkcolor=LinksColor,urlcolor=LinksColor,citecolor=LinksColor}
\newcommand{\lref}[2][]{\hyperref[#2]{#1~\ref*{#2}}}
\usepackage{fullpage}
\usepackage{scrextend}

\newcommand{\lastcorrections}%
{{
 \begin{sloppypar}
    \baselineskip -0.2in
    \tiny\bf\noindent
last corrections:\\
\end{sloppypar}
}}

\newcommand{\margincomment}[1]%
    {{%
      \marginpar{{\tiny\begin{minipage}{0.5in}
                       \begin{flushleft}
                          {#1}
                       \end{flushleft}
                       \end{minipage}
                }}
    }}

\newcommand{\ignore}[1]{}

\newcommand{\emparagraph}[1]{{\smallskip\noindent\emph{#1}}}

\newcommand{\bard}{\bar{d}}
\newcommand{\barn}{\bar{n}}

\newcommand{\bars}{\bar{s}}

\newcommand{\barV}{{\bar{V}}}
\newcommand{\barZ}{{\bar{Z}}}

\newcommand{\bfE}{{\bf{E}}}
\newcommand{\bfP}{{\bf{P}}}

\newcommand{\calA}{{\cal A}}
\newcommand{\calB}{{\cal B}}

\newcommand{\calH}{{\cal H}}

\newcommand{\calP}{{\cal P}}
\newcommand{\calR}{{\cal R}}

\newcommand{\calY}{{\cal Y}}
\newcommand{\calZ}{{\cal Z}}
\newcommand{\calT}{{\cal T}}

\newcommand{\tilden}{{\tilde{n}}}
\newcommand{\tildeV}{{\tilde{V}}}
\newcommand{\tildeZ}{{\tilde{Z}}}

\newcommand{\braced}[1]{{ \left\{ #1 \right\} }}

\newcommand{\ceiling}[1]{{ \lceil #1 \rceil }}
\newcommand{\floor}[1]{{ \lfloor #1 \rfloor }}

\newcommand{\suchthat}{{\;:\;}}

\newcommand{\integers}{{\mathbb Z}}

\newcommand{\polylog}{{\mbox{\rm polylog}}}

\newcommand{\height}{{\textit{height}}}
\newcommand{\depth}{{\textit{depth}}}
\newcommand{\Depth}{D}

\newcommand{\algSimAgDTree}{\mbox{\sc{UnbDTree1}}}
\newcommand{\algLinAgDTree}{\mbox{\sc{UnbDTree2}}}
\newcommand{\algNoAgDTree}{\mbox{\sc{BndDTree}}}
\newcommand{\algMlessDTree}{\mbox{\sc{MlsDTree}}}

\newcommand{\algRandTree}{\mbox{\sc RTree}}
\newcommand{\vlabel}{{\textsf{label}}}

\newcommand{\myAll}{{\textsf{All}}}
\newcommand{\RoundRobin}{{\textsc{RoundRobin}}}
\newcommand{\myRR}{{\textsf{RR}}}
\newcommand{\Selector}[1]{{\textsc{$#1$-Select}}}
\newcommand{\mySel}{{\textsf{Sel}}}

\newtheorem{theorem}{Theorem}

\newtheorem{lemma}{Lemma}

\newtheorem{claim}{Claim}

\newenvironment{myalgorithm}[1]{ \medskip \begin{addmargin}{0.3in} \noindent\hrulefill%
						\medskip\newline\noindent{\bf{Algorithm~{#1}.}}}
						{ \newline\smallskip\noindent\hrulefill\end{addmargin}\medskip }

		{$\spadesuit$\smallskip}

\newenvironment{bigeqn*}{\large\begin{eqnarray*}}{\end{eqnarray*}}

\newcommand{\half}{{\textstyle\frac{1}{2}}}

\begin{document}

\title{Information Gathering in Ad-Hoc Radio Networks\\ with Tree Topology}

\author{Marek Chrobak\thanks{%
			Department of Computer Science,
         	University of California at Riverside, USA.
			Research supported by NSF grants CCF-1217314 and OISE-1157129.
			}
		\and
		 Kevin Costello\thanks{%
			Department of Mathematics,
			University of California at Riverside, USA.
            Research supported by NSA grant H98230-13-1-0228
			}
		\and
		Leszek Gasieniec\thanks{%
			Department of Computer Science,
			University of Liverpool, UK.
			}
		\and
		Dariusz R. Kowalski\footnotemark[3]
}

\maketitle

\begin{abstract}
We study the problem of information gathering in ad-hoc radio networks without collision detection,
focussing on the case when the network forms a tree, with edges directed towards the root.
Initially, each node has a piece of information that we
refer to as a rumor. Our goal is to design protocols that deliver all rumors to the
root of the tree as quickly as possible. The protocol must complete this
task within its allotted time even though the actual tree topology is unknown
when the computation starts. In the deterministic case, assuming that
the nodes are labeled with small integers, we give
an $O(n)$-time protocol that uses unbounded messages,
and an $O(n\log n)$-time protocol using bounded messages, where any message
can include only one rumor.
We also consider fire-and-forward protocols, in which a node can only transmit
its own rumor or the rumor received in the previous step. 
We give a deterministic fire-and-forward protocol with
running time $O(n^{1.5})$, and we show that it is asymptotically optimal.
We then study randomized algorithms where the nodes are not labelled.
In this model, we give an $O(n\log n)$-time protocol and we
prove that this bound is asymptotically optimal.
\end{abstract}


\section{Introduction}
\label{sec: introduction}



We consider the problem of information gathering in ad-hoc radio networks,
where initially each node has a piece of information called a \emph{rumor},
and all these rumors need to be delivered to a designated target node
as quickly as possible.
A radio network is defined as a directed graph $G$ with $n$ vertices.
At each time step any node $v$ of $G$ may attempt to transmit a message. This
message is sent immediately to all out-neighbors of $v$. However, an out-neighbor $u$ of $v$
will receive this message only if no other in-neighbor of $u$ attempted to
transmit at the same step. The event when two or more in-neighbors of $u$ transmit
at the same time is called a \emph{collision}.  We do not assume
any collision detection mechanism; in other words, not only this $u$ will not
receive any message, but it will not even know that a collision occurred.

One other crucial feature of our model is that the topology of $G$ is
not known at the beginning of computation. We are interested in distributed 
protocols, where the execution of a protocol at a node $v$ depends only
on the identifier (label) of $v$ and the information gathered from
the received messages. Randomized protocols typically do not
use the node labels, and thus they work even if the nodes are
indistinguishable from each other.
The protocol needs to complete its task within the allotted time, 
independently of the topology of $G$.

Several primitives for information dissemination in ad-hoc radio networks have
been considered in the literature. Among these, 
the two most extensively studied are \emph{broadcasting} and \emph{gossiping}.

The \emph{broadcasting problem} is the one-to-all dissemination problem, where
initially only one node has a rumor that needs to be
delivered to all nodes in the network. Assuming that the nodes of $G$
are labelled with consecutive integers $0,1,...,n-1$, 
the fastest known deterministic algorithms for broadcasting
run in time $O(n\log n\log\log n)$~\cite{DeMarco_08} or
$O(n\log^2D)$ \cite{Czumaj_Rytter_broadcast_06}, where $D$ is the diameter of $G$.
The best lower bound on the running time in this model is $\Omega(n\log D)$
\cite{Clementi_etal_distributed_02}. 
(See also \cite{Chrobak_etal_fast_02,Kowalski_Pelc_faster_04,Bruschi_etal_lower_bound_97,Chlebus_etal_broadcast_02} 
for earlier work.)
Allowing randomization, broadcasting can be accomplished in time $O(D\log(n/D)+\log^2n)$ with high probability
\cite{Czumaj_Rytter_broadcast_06}, even if the nodes are not labelled.
This matches the lower bounds in \cite{alon_etal_lower_bound_91,Kushilevitz_Mansour_lower_bound_98}.

The \emph{gossiping problem} is the all-to-all dissemination problem. Here,
each node starts with its own rumor and the goal is to deliver all rumors to each node.
There is no restriction on the size of messages; in particular,
different rumors can be transmitted together in a single message.
With randomization, gossiping can be solved in expected
time $O(n\log^2n)$ \cite{Czumaj_Rytter_broadcast_06} (see
\cite{Liu_Prabhakaran_randomized_02,Chrobak_etal_randomized_04} for earlier work),
even if the nodes are not labelled. 
In contrast, for deterministic algorithms, with nodes labelled $0,1,...,n-1$,
the fastest known gossiping algorithm runs in time
$O(n^{4/3}\log^4n)$ \cite{Gasieniec_etal_det_gossip_04},
following earlier progress in \cite{Chrobak_etal_fast_02,Xu_det_gossip_03}.
(See also the survey in \cite{Gasieniec_survey_gossiping_09} for more information.)
For graphs with arbitrary diameter,
the best known lower bound is $\Omega(n\log n)$, the same as for broadcasting.
Reducing the gap between lower and upper bounds for deterministic gossiping
to a poly-logarithmic factor remains a central open problem in
the study of radio networks with unknown topology.

Our work has been inspired by this open problem. It is easy to see that
for arbitrary directed graphs gossiping is equivalent to information
gathering, in the following sense. On one hand, trivially, any protocol for
gossiping also solves the problem of gathering. On the other hand, we can
apply a gathering protocol and follow it with a protocol that broadcasts all
information from the target node $r$; these two protocols combined
solve the problem of gossiping. So if we can solve information gathering in time $O(n\,\polylog(n))$,
then we can also solve gossiping in time $O(n\,\polylog(n))$.


\paragraph{Our results.}
To gain better insight into the problem of gathering information in radio
networks, we focus on networks with tree topology. Thus we assume that
our graph is a tree $\calT$ with root $r$ and with all edges directed
towards $r$. In this model, a gathering protocol knows that the network
is a tree, but it does not know its topology. 

We consider several variants of this problem, for deterministic or randomized algorithms,
and with or without restrictions on the message size or processor memory. 
We provide the following results:
\begin{itemize}
\item
In the first part of the paper we study deterministic algorithms, under the assumption that
the nodes of $\calT$ are labelled $0,1,...,n-1$. 
First, in Section~\ref{sec: det algorithms with aggregation}, we examine the 
model without any bound on the message size. In particular, such protocols are allowed to aggregate
any number of rumors into a single message. (This is a standard model in existing gossiping 
protocols in unknown radio networks; see, for example, the survey in 
\cite{Gasieniec_survey_gossiping_09}.)
We give an optimal, $O(n)$-time protocol using unbounded messages.
\item
Next, in Section~\ref{sec: deterministic without aggregation},
we consider the model with bounded messages, where a message may contain
only one rumor. For this model we provide
an algorithm with running time $O(n\log n)$.
\item
In Section~\ref{sec: det fireandforward protocols}
we consider an even more restrictive model of protocols with bounded messages,
that we call fire-and-forward protocols. In those protocols, at each step, a node
can only transmit either its own rumor or the rumor received in the previous step (if any).
For deterministic fire-and-forward protocols we provide a protocol with running time
$O(n^{1.5})$ and we show a matching lower bound of $\Omega(n^{1.5})$.
\item
We then turn our attention to randomized algorithms 
(Sections~\ref{sec: nlogn randomized} and~\ref{sec: nlogn randomized lower bound}).
For randomized algorithms we assume that the nodes are not labelled.
In this model, we give an $O(n\log n)$-time gathering protocol and we prove a matching lower bound
of $\Omega(n\log n)$.
The upper bound is achieved by a simple fire-and-forward protocol that, 
in essence, reduces the problem to the coupon-collector problem. Our main
contribution here is the lower bound proof. For the special case of trees of depth two, we 
show that our
lower bound is in fact optimal even with respect to the leading constant. 
\end{itemize}

All our algorithms for deterministic protocols easily extend to the model
where the labels are drawn from a set $0,1,...,L$ where $L = O(n)$, without
affecting the running times. If $L$ is arbitrary, the
algorithms for unbounded and bounded message models can be
implemented in times, respectively, $O(n^2\log L)$ and $O(n^2\log n\log L)$.

We remark that some communication protocols for radio networks
use forms of information gathering on trees as a sub-routine; see for example
\cite{Bar-Yehuda_I_I_93,chlebus_K_R_09,Khabbazian_kowalski_11}.
However, these solutions typically focus on undirected graphs, which allow feedback.
They also solve relaxed variants of information gathering where the goal is to gather
only a fraction of rumors in the root, which was sufficient for the applications
studied in these papers (since, with feedback, such a procedure can be repeated until
all rumors are collected).
In contrast, in our work, we study directed trees without any feedback mechanism, and
we require all rumors to be collected at the root.

Our work leaves several open problems, with the most intriguing one being whether
there is a deterministic algorithm for trees that does not use aggregation and works in time
$o(n\log n)$. We conjecture that such an algorithm exists. 
It would also be interesting to refine our results by expressing the 
running time in terms of the tree depth $D$ and the maximum degree $\Delta$.
In terms of more general research directions, the next goal should be
to establish some bounds for gathering in arbitrary acyclic graphs.


\section{Preliminaries}
\label{sec: preliminaries}



\paragraph{Radio networks.}
In this section we give formal definitions of radio networks and gathering protocols.
We define a radio network as a directed graph $G = (V,E)$ with $n$ nodes,
with each node assigned a different label from  the set $[n] = \braced{0,1,...,n-1}$.
Denote by $\vlabel(v)$ the label assigned to a node $v\in V$.
One node $r$ is distinguished as the \emph{target} node, and we assume that $r$
is reachable from all other nodes.
Initially, at time $0$, each node $v$ has some piece of information that
we will refer to as \emph{rumor} and we will denote it by $\rho_v$. 
The objective is to deliver all rumors $\rho_v$ to $r$ as quickly as
possible, according to the rules described below.

The time is discrete, namely it consists of time steps numbered with
non-negative integers $0,1,2,...$.
At any step, a node $v$ may be either in the \emph{transmit state} or 
the \emph{receive state}.
A gathering protocol $\calA$ determines, for each node $v$ and
each time step $t$, whether $v$ is in the transmit or receive state at time $t$. If 
$v$ is in the transmitting state, then $\calA$ also determines what message
is transmitted by $v$, if any. This specification of $\calA$
may depend only on the label of $v$, time $t$, 
and on the content of all messages received by $v$ until time $t$.
We stress that, with these restrictions, $\calA$ does not depend on the
topology of $G$ and on the node labeling. 

All nodes start executing the protocol simultaneously at time $0$. 
If a node $v$ transmits at a time $t$, 
the transmitted message is sent immediately to all out-neighbors of $v$, that is to all $u$ such that
$(v,u)$ is an edge. If $(v,u)$ and $(v',u)$ are edges and both $v,v'$ transmit at time $t$
then a \emph{collision} at $u$ occurs and $u$ does not receive a message. 
More specifically, $u$ will receive this message if and only if 
(i) $u$ is in the receive state at time $t$, and (ii)  no collision at $u$ occurs at time $t$.
We do not assume any feedback from the transmission channel or any collision detection features, so,
in case of a collision,
neither the sender nor any node within its range knows that a collision occurred.

Throughout the paper, we will focus on the case when the graph is a tree with root $r$ and
with all edges directed towards the root $r$. We will
typically use notation $\calT$ for this tree. 

The running time of a deterministic gathering protocol $\calA$ is defined as the minimum time
$T(n)$ such that, for any tree $\calT$ with root $r$ and $n$ nodes, any assignment of labels
from $[n]$ to the nodes of $\calT$, and any node $v$,
the rumor $\rho_v$ of $v$ is delivered to $r$ no later than at step $T(n)$. In case of
randomized protocols, we evaluate them either using the expectation of their
running time $T(n)$, which is now a random variable, or by showing that $T(n)$ does
not exceed a desired time bound with high probability.

\smallskip

We consider three types of gathering protocols:
\begin{description}
	\item{\emph{Unbounded messages:}} In this model a node can transmit
			arbitrary information in a single step.  In particular, multiple
			rumors can be aggregated into a single message.
	\item{\emph{Bounded messages:}} In this model no aggregation of rumors
			is allowed. Each message consists of at most one rumor and 
			$O(\log n)$ bits of additional information.
	\item{\emph{Fire-and-forward:}} In a fire-and-forward protocol, a
	node can either transmit its own rumor or the rumor received in the
	previous step, if any. Thus a message originating from
	a node travels towards the root one hop at a time, until either it vanishes
	(due to collision, or being dropped by a node that fires),
	or it successfully reaches the root. 
\end{description}
For illustration, consider a protocol called $\RoundRobin$, where
all nodes transmit in a cyclic order, one at a time.
Specifically, at any step $t$, $\RoundRobin$ transmits from the node $v$ 
with $\vlabel(v) = t\bmod{n}$, with 
the transmitted message containing all rumors received by $v$ until time $t$.
The running time is $O(n^2)$, because initially each rumor $\rho_u$ is at distance at most
$n$ from the root and in any consecutive $n$ steps it will decrease its distance to the 
root by at least $1$. ($\RoundRobin$ has been used as a subroutine in many protocols in
the literature, and it achieves running time $O(n^2)$ even for gossiping in arbitrary
networks, not only for gathering in trees.)

$\RoundRobin$ can be adapted to use only bounded messages for information gathering
in trees. At any round $t$ and any node $v$, if $v$ has the rumor $\rho_u$ of node $u$
such that $\vlabel(u) = t\bmod{n}$, and $v$ has not transmitted $\rho_u$ before, then
$v$ transmits $\rho_u$ at time $t$. Since $\calT$ is a tree, each rumor follows
the unique path towards the root, so no collisions will occur,
and after at most $n^2$ steps $r$ will receive all rumors.


\section{Some Structure Properties of Trees}
\label{sec: some structure properties of trees}



The running times of our algorithms in Sections~\ref{sec: det algorithms with aggregation}
and~\ref{sec: deterministic without aggregation} depend on the
distribution of high-degree nodes within a tree. To capture the structure
of this distribution we use the concept of $\gamma$-depth 
which, roughly, measures how ``bushy'' the tree is. We define this
concept in this section and establish its properties needed for the analysis of
our algorithms.


\paragraph{$\gamma$-Depth of trees.}
Let $\calT$ be the given tree network with root $r$ and $n$ nodes. Fix an integer
$\gamma$ in the range $2\le \gamma\le n-1$.
We define the \emph{$\gamma$-height} of each node $v$ of $\calT$, denoted $\height_\gamma(v)$,
as follows. If $v$ is a leaf then $\height_\gamma(v) = 0$.
If $v$ is an internal node then let $g$ be the maximum $\gamma$-height of a child
of $v$. If at least $\gamma$ children of $v$ have $\gamma$-height equal $g$
then $\height_\gamma(v) = g+1$; otherwise $\height_\gamma(v) = g$.
(For $\gamma = 2$, our definition of $2$-height is equivalent to that
of Strahler numbers. See, for example, \cite{Strahler_52,Viennot_02}.)
We then define the \emph{$\gamma$-depth of $\calT$} as
$\Depth_\gamma(\calT) = \height_\gamma(r)$.

In our proofs, we may also consider trees other than the input tree $\calT$.
If $\calH$ is any tree and $v$ is a node of $\calH$ then, to avoid
ambiguity, we will write $\height_\gamma(v,\calH)$ for the $\gamma$-height of $v$ with
respect to $\calH$. Note that if $\calH$ 
is a subtree of $\calT$ and $v\in \calH$ then, trivially,
$\height_\gamma(v,\calH) \le \height_\gamma(v)$.

By definition, the $1$-height of a node is the same as its height, namely the longest distance
from this node to a leaf in its subtree. For a tree, its $1$-depth is equal to its depth.
Figure~\ref{fig: gamma-height example} shows an example of a tree whose depth equals $4$,
$2$-depth equals $3$, and $3$-depth equals $1$.

\begin{figure}[ht]
\begin{center}
\includegraphics[width=4in]{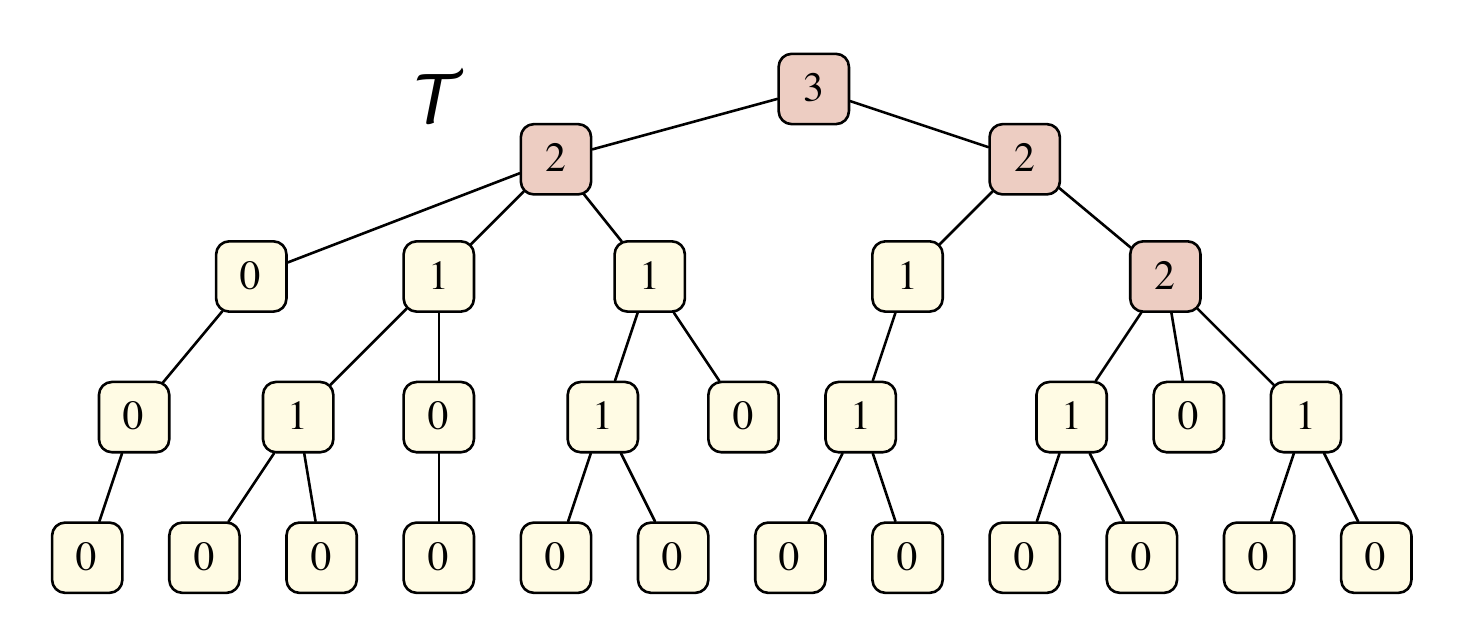}
\caption{An example illustrating the concept of $\gamma$-depth of trees, for $\gamma = 1,2,3$. 
The depth of this tree $\calT$ is $4$.
The number in each node is its $2$-height; thus the $2$-depth of this tree is $3$. 
All light-shaded nodes have $3$-height equal $0$ and the four
dark-shaded nodes have $3$-height equal $1$, so the $3$-depth of this tree is $1$.}
\label{fig: gamma-height example}
\end{center}
\end{figure}

The lemmas below spell out some simple properties of $\gamma$-heights of nodes
that will be useful for the analysis of our algorithms. (In particular,
Lemma~\ref{lem: gamma-heights} generalizes the bound from \cite{Viennot_02} for
$\gamma =2$.)
If $v$ is a node of $\calT$ then $\calT_v$ will denote
the subtree of $\calT$ rooted at $v$ and containing all descendants of $v$.


\begin{lemma}\label{lem: gamma-heights}
$\Depth_\gamma(\calT) \le \log_\gamma n$.
\end{lemma}

\begin{proof}
It is sufficient to show that $|\calT_v| \ge \gamma^{\height_\gamma(v)}$ holds for each node $v$.
The proof is by simple induction with respect to the height of $v$.

If $v$ is a leaf then the inequality is trivial. So suppose now that
$v$ is an internal node with $\height_\gamma(v) = g$.
If $v$ has a child $u$ with $\height_\gamma(u) = g$ then, by induction,
$|\calT_v| \ge |\calT_u| \ge \gamma^g$.
If all children of $v$ have $\gamma$-height smaller than $g$ then
$v$ must have at least $\gamma$ children with $\gamma$-height equal
$g-1$. So, by induction, we get
$|\calT_v|\ge \gamma\cdot \gamma^{g-1} = \gamma^g$.
\end{proof}

We will be particularly interested in subtrees of $\calT$ consisting of the nodes
whose $\gamma$-height is above a given threshold. Specifically, for
$h = 0,1,...,\Depth_\gamma(\calT)$, let
$\calT^{\gamma,h}$ be the subtree of $\calT$ induced by the nodes
whose $\gamma$-height is at least $h$
(see Figure~\ref{fig: tree reduced}). Note that, since $\gamma$-heights
are monotonically non-decreasing on the paths from leaves to $r$,
$\calT^{\gamma,h}$ is indeed a subtree of $\calT$ rooted at $r$.
In particular, for $h = 0$ we have $\calT^{\gamma,0} = \calT$.

For any $h$, $\calT - \calT^{\gamma,h}$ is a collection of subtrees
of type $\calT_v$, where $v$ is a node of $\gamma$-height less than $h$ whose
parent is in $\calT^{\gamma,h}$.
When $h=1$, all such subtrees contain only nodes of $\gamma$-height equal $0$, which
implies that they all have degree less than $\gamma$.
In particular, for $\gamma = 2$, each such subtree $\calT_v$ is a path from a leaf of
$\calT$ to $v$.

\begin{figure}[ht]
\begin{center}
\includegraphics[width=2.5in]{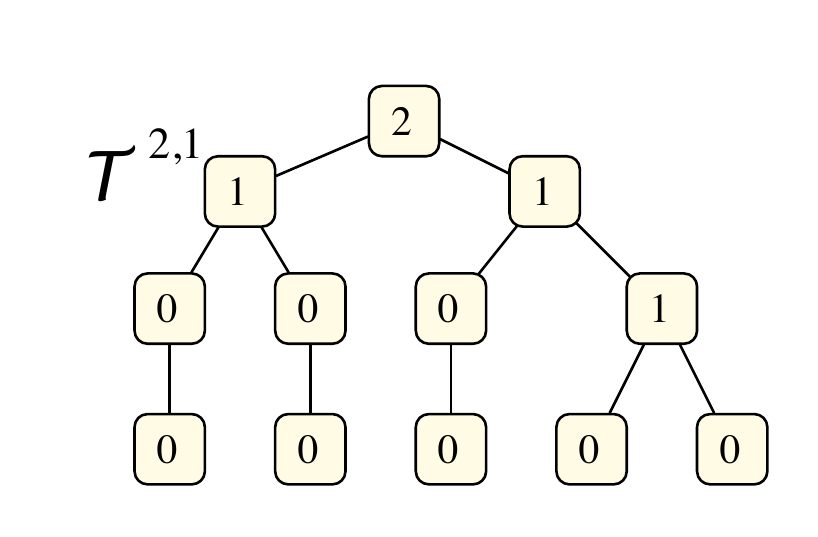}
\caption{The subtree $\calT^{2,1}$ obtained from tree $\calT$ in Figure~\ref{fig: gamma-height example}.}
\label{fig: tree reduced}
\end{center}
\end{figure}

 
\begin{lemma}\label{lem: subtrees gamma-heights}
For any node $v\in\calT^{\gamma,h}$ we have
$\height_\gamma(v,\calT^{\gamma,h}) = \height_\gamma(v) - h$.
Thus, in particular, we also have $\Depth_\gamma(\calT^{\gamma,h}) = \Depth_\gamma(\calT)-h$.
\end{lemma}

\begin{proof}
It is sufficient to prove the lemma for the case $h=1$, which can be then extended to
arbitrary values of $h$ by induction.

So let $h =1$ and $\calT' = \calT^{\gamma,1}$.
The proof is by induction on the height of $v$ in $\calT'$.
If $v$ is a leaf of $\calT'$ then $\height_\gamma(v,\calT') = 0$, by definition.
All children of $v$ in $\calT$ must have $\gamma$-height
equal $0$, so $\height_\gamma(v) = 1$ and thus the lemma holds for $v$.

Suppose now that $v$ is not a leaf of $\calT'$ and that the lemma holds for all children of $v$.
This means that for each child $u$ of $v$ in $\calT$, either  $\height_\gamma(u)=0$ 
(that is, $u\notin \calT'$) or $\height_\gamma(u,\calT') = \height_\gamma(u)-1$.
Let $\height_\gamma(v) = f$.
If $v$ has a child with $\gamma$-height equal $f$ then there are fewer than $\gamma$
such children. By induction, these children will have $\gamma$-height in $\calT'$
equal $f-1$, and each other child that remains in $\calT'$
has $\gamma$-height in $\calT'$ smaller than $f-1$.
So $\height_\gamma(v,\calT') = f-1$.
If all children of $v$ have $\gamma$-height smaller than $f$ then 
$f\ge 2$ (for otherwise $v$ would have to be a leaf of $\calT'$) and
$v$ must have more than $\gamma$ children with $\gamma$-height $f-1$. 
These children will be in $\calT'$ and will have
$\gamma$-height in $\calT'$ equal $f-2$. So $\height_\gamma(v,\calT') = f-1$ in this case as well,
completing the proof.
\end{proof}


\section{Deterministic Algorithms with Aggregation}
\label{sec: det algorithms with aggregation}



As explained in Section~\ref{sec: preliminaries}, it is easy to achieve
information gathering in time $O(n^2)$. In this section
we prove that using unbounded-size messages this running time can be improved to $O(n)$. 
This is optimal for arbitrary $n$-node trees, since no better time can be achieved 
even for paths or star graphs. We start with a simpler protocol with running time $O(n\log n)$, 
that we use to introduce our terminology and techniques. 

We can make some assumptions about the protocols in this section that will simplify their
presentation.  Since we use unbounded messages, we can assume that each 
transmitted message contains all information received by the transmitting node,
including all received rumors. We also assume that all
rumors are different, so that each node can keep track of the number
of rumors collected from its subtree. To ensure this we can,
for example, have each node $v$ append its label to its rumor $\rho_v$.

We will also assume that each node knows the labels of its children. To acquire
this knowledge, we can precede any protocol by a preprocessing phase where
nodes with labels $0,1,...,n-1$ transmit, one at a time, in this order. 
Thus after $n$ steps each node will receive the messages from its children. 
This does not affect the asymptotic running times of our protocols.



\subsection{Warmup: a Simple $O(n\log n)$-Time Algorithm}
\label{sec: simple det with aggregation}



We now present an algorithm for information gathering on trees that runs in time $O(n\log n)$. 
In essence, any node waits until it receives the messages from its children,
then for $2n$ steps it alternates $\RoundRobin$ steps with steps when it always attempts to transmit.
A more detailed specification of the algorithm follows.


\begin{myalgorithm}{$\algSimAgDTree$}
We divide the time steps into \emph{rounds}, where round $s$ consists of two consecutive steps $2s$ and $2s+1$,
which we call, respectively, the $\myRR$-step and the $\myAll$-step of round $s$. 

For each node $v$ we define its \emph{activation round}, denoted $\alpha_v$, as follows.
If $v$ is a leaf then $\alpha_v = 0$. For
any other node $v$, $\alpha_v$ is the first round such that $v$ has 
received messages from all its children when this round is about to start.

For each round $s = \alpha_v,\alpha_v+1, ...,\alpha_v+n-1$, 
$v$ transmits in the $\myAll$-step of round $s$,
and if $\vlabel(v) = s\bmod{n}$ then it also transmits
in the $\myRR$-step of round $s$.
In all other steps, $v$ stays in the receiving state. 
\end{myalgorithm}


\medskip

\emparagraph{Analysis.}
For any node $v$ we say that $v$ is \emph{dormant} in rounds $0,1,...,\alpha_v-1$,
$v$ is \emph{active} in rounds $\alpha_v,\alpha_v+1,...,\alpha_v+n-1$,
and that $v$ is \emph{retired} in every round thereafter. 
Since $v$ will make at least one
$\myRR$-transmission when it is active, $v$ will successfully transmit its message
(and thus all rumors from its subtree $\calT_v$) to its parent before retiring, and before
this parent is activated.
Therefore, by a simple inductive argument, Algorithm~$\algSimAgDTree$
is correct, namely that eventually $r$ will receive all rumors from $\calT$.

This inductive argument shows in fact that, at any round, Algorithm~$\algSimAgDTree$
satisfies the following two invariants:
(i) any path in $\calT$ from a leaf to $r$ consists of a segment of retired nodes,
followed by a segment of active nodes, which is then followed by a segment of
dormant nodes (each of these segments possibly empty); and
(ii) any dormant node has at least one active descendant.
 

\begin{lemma}\label{lem: nlogn height h}
Let $d = \Depth_2(\calT)$. For any $h = 0,1,...,d$ and
any node $v$ with $\height_2(v) = h$, $v$ gets activated no later than
in round $2nh$, that is $\alpha_v \le 2nh$.
\end{lemma}

\begin{proof}
The proof is by induction on $h$. By the algorithm, the lemma
trivially holds for $h=0$.
Suppose that the lemma holds for $h-1$ and consider a node $v$
with $\height_2(v) = h$. To reduce clutter, denote
$\calZ = \calT^{2,h}$. From Lemma~\ref{lem: subtrees gamma-heights},
we have that $\height_2(v,\calZ) = 0$, which implies that
$\calZ_v$ is a path from a leaf of $\calZ$ to $v$.
Let $\calZ_v = v_1,v_2,...,v_q$ be this path, where $v_1$ is a leaf of $\calZ$ and $v_q = v$.

We now consider round $s = 2n(h-1)+n$.
The nodes in $\calT-\calZ$ have $2$-height at most $h-1$, so,
by the inductive assumption, they are activated 
no later than in round $2n(h-1)$, and therefore in round $s$ they are already retired.
If $\alpha_v \le s$ then $\alpha_v \le 2nh$, and we are done. 
Otherwise, $v$ is dormant in round $s$. Then, by invariant (ii) above, 
at least one node in $\calZ_v$ must be active. Choose the largest
$p$ for which $v_p$ is active in round $s$. In round $s$ and later, all children of
the nodes $v_p,v_{p+1},...,v_q$ that are not on $\calZ_v$ do not transmit, since
they are already retired. This implies that for each $\ell = 0,...,q-p-1$,
node $v_{p+\ell+1}$ will get activated in round $s+\ell+1$
as a result of the $\myAll$-transmission from node $v_{p+\ell}$.
In particular, we obtain that 
$\alpha_v \le s+q-p \le  2nh$, completing the proof of the lemma.
\end{proof}

We have $\height_2(r) = d$ and $d = O(\log n)$, by Lemma~\ref{lem: gamma-heights}.
Applying Lemma~\ref{lem: nlogn height h}, this implies 
that $\alpha_r \le 2nd = O(n\log n)$, 
which gives us that the overall running time is $O(n\log n)$. 
Summarizing, we obtain the following theorem.


\begin{theorem}\label{thm: nlog time with aggregation}
For any tree with $n$ nodes and any assignment of labels,
Algorithm~$\algSimAgDTree$ completes information gathering in time $O(n\log n)$.
\end{theorem}


\subsection{An $O(n)$-Time Deterministic Algorithm}
\label{sec: linear deterministic with aggregation}



In this section we show how to improve the running time of information gathering
in trees to linear time, assuming unbounded size messages.
The basic idea is to use strong $k$-selective families to speed up the computation.

Recall that a \emph{strong $k$-selective family}, where $1\le k\le n$,
is a collection 
$F_0,F_1,...,F_{m-1} \subseteq [n]$ of sets such that for any set $X\subseteq [n]$ 
with $|X| \le k$ and any
$x\in X$, there is $j$ for which $F_j\cap X = \braced{x}$.
It is well known that for any $k=1,2,...,n$, there is a strong $k$-selective family
with $m = O(k^2\log n)$ sets~\cite{Erdos_etal_families_85,Clementi_etal_distributed_02}.
Note that in the special case $k=1$ the family consisting of just one set $F_0 = [n]$
is $1$-selective. This corresponds to $\myAll$-transmissions in the previous section.
For $k = \omega(\sqrt{n/\log n})$ we can also improve the $O(k^2\log n)$ bound to $O(n)$
by using the set family corresponding to the $\RoundRobin$ protocol, namely the
$n$ singleton sets $\braced{0}, \braced{1},...,\braced{n-1}$.

In essence, the strong $k$-selective family can be used to speed up information 
dissemination through low-degree nodes. Consider  
the protocol $\Selector{k}$ that works as follows: for any step $t$ and any node $v$, if
$\vlabel(v) \in F_{t\bmod{n}}$ then transmit from $v$, otherwise stay in the receive state. 
Suppose that $w$ is a node with fewer than $k$ children, and that these children
collected the messages from their subtrees by time $t$. Steps $t,t+1,...,t+m-1$
of this protocol use all sets $F_0,F_1,...,F_{m-1}$ 
(although possibly in a different order), so each child of $w$ will make a successful
transmission by time $t+m$, which is faster than time $O(n)$ required by
$\RoundRobin$ if $k = o(\sqrt{n/\log n})$.

To achieve linear time for arbitrary trees, we will interleave the steps of
protocol $\Selector{k}$ with $\RoundRobin$ (to deal with high-degree nodes) and
steps where all active nodes transmit (to deal with long paths).

Below, we fix parameters $\kappa = \ceiling{n^{1/3}}$ and $m = O(\kappa^2\log n)$,
the size of a strong $\kappa$-selective family $F_0,F_1,...,F_{m-1}$. (The choice of
$\kappa$ is somewhat arbitrary; in fact, any $\kappa = \Theta(n^c)$, for $0 < c < \half$, would work.)
Without loss of generality
we can assume that $m \le n$.


\begin{myalgorithm}{$\algLinAgDTree$}
We divide the steps into rounds, where each round $s$ consists of three
consecutive steps $3s$, $3s+1$, and $3s+2$, that we will call the
$\myRR$-step, $\myAll$-step, and $\mySel$-step of round $s$, respectively.

For each node $v$ we define its \emph{activation round}, denoted $\alpha_v$, as follows.
If $v$ is a leaf then $\alpha_v = 0$. For
any other node $v$, $\alpha_v$ is the first round such that before this round starts
$v$ has received all messages from its children.

In each round $s = \alpha_v, \alpha_v+1,...,\alpha_v+m-1$, 
$v$ transmits in the $\myAll$-step of round $s$,
and if $\vlabel(v) \in F_{s\bmod{m}}$ then $v$ also transmits in the $\mySel$-step
of round $s$.
In each round $s = \alpha_v, \alpha_v+1,...,\alpha_v+n-1$, 
if $\vlabel(v) = s\bmod{n}$ then $v$ transmits in the $\myRR$-step of
round $s$.
If $v$ does not transmit according to the above rules then
$v$ stays in the receiving state.
\end{myalgorithm}


\emparagraph{Analysis.}
Similar to Algorithm~$\algSimAgDTree$, in Algorithm~$\algLinAgDTree$ each
node $v$ goes through three stages. We call $v$ \emph{dormant}
in rounds $0,1,...,\alpha_v-1$, \emph{active} in rounds 
$\alpha_v,\alpha_v+1,...,\alpha_v+n-1$, and retired thereafter. We will also refer to $v$
as being \emph{semi-retired} in rounds $\alpha_v+m,\alpha_v+m+1,...,\alpha_v+n-1$ (when it is
still active, but only uses $\myRR$-transmissions).
Assuming that $v$ gets activated in some round, 
since $v$ makes at least one $\myRR$-transmission when it is active, it
will successfully transmit its message to its parent before retiring, and
before its parent gets activated. By straightforward
induction on the depth of $\calT$, this implies that
each node will eventually get activated, proving that
Algorithm~$\algLinAgDTree$ is correct.

By a similar argument, Algorithm~$\algLinAgDTree$ satisfies the following
two invariants in each round:
(i) Any path from a leaf to $r$ consists of a segment of retired nodes,
followed by a segment of active nodes (among the active nodes, the
semi-retired nodes precede those that are not semi-retired), 
which is then followed by a segment of dormant nodes.
(ii) Any dormant node has at least one active descendant.

It remains to show that the running time of Algorithm~$\algLinAgDTree$ is $O(n)$.
The idea of the analysis is to show that 
$\mySel$- and $\myAll$-steps disseminate information very fast, in linear time,
through subtrees where all node degrees are less than $\kappa$. (In fact, this
applies also to nodes with higher degrees, as long as they have fewer than
$\kappa$ active children left.)
The process can stall, however, if all active nodes have parents of degree
larger than $\kappa$. In this case, a complete cycle of $\RoundRobin$ will
transmit the messages from these nodes to their parents. We show, using Lemma~\ref{lem: gamma-heights},
that, since $\kappa =  \ceiling{n^{1/3}}$, such stalling can occur at most the total of $3$ times.
So the overall running time 
will be still $O(n)$. We formalize this argument in the remainder of this sub-section.

Let $\bard = \Depth_\kappa(\calT)$.
From Lemma~\ref{lem: gamma-heights}, we have $\bard\le 3$.
We fix some integer $g\in\braced{0,1,2,3}$, a node $w$ with
$\height_\kappa(w) = g$, and we let $\calY = \calT^{\kappa,g}_w$. Thus
$\calY$ consists of the descendants of $w$ (including $w$ itself) whose 
$\kappa$-height in $\calT$ is exactly $g$, or, equivalently (by Lemma~\ref{lem: subtrees gamma-heights}), the
descendants of $w$ in $\calT^{\kappa,g}$ whose $\kappa$-height in
$\calT^{\kappa,g}$ is equal $0$. (See Figure~\ref{fig: linear tree Y} for illustration.)
Therefore all nodes in $\calY$ have degree strictly smaller than $\kappa$.

We also fix $\bars$ to be the first round when all nodes in
$\calT - \calT^{\kappa,g}$ are active or already retired. In particular, 
for $g=0$ we have $\bars = 0$. Our goal now is to show
that $w$ will get activated in at most $O(n)$ rounds after round $\bars$.

\begin{figure}[ht]
\begin{center}
\includegraphics[width=2.4in]{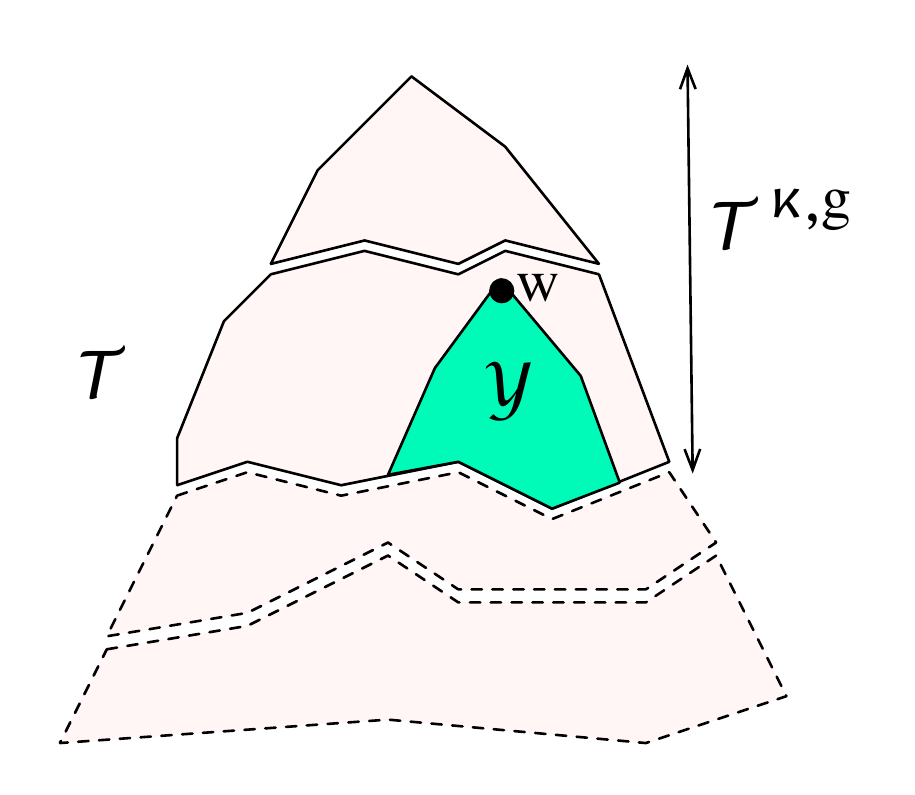}
\ \ \
\includegraphics[width=2.4in]{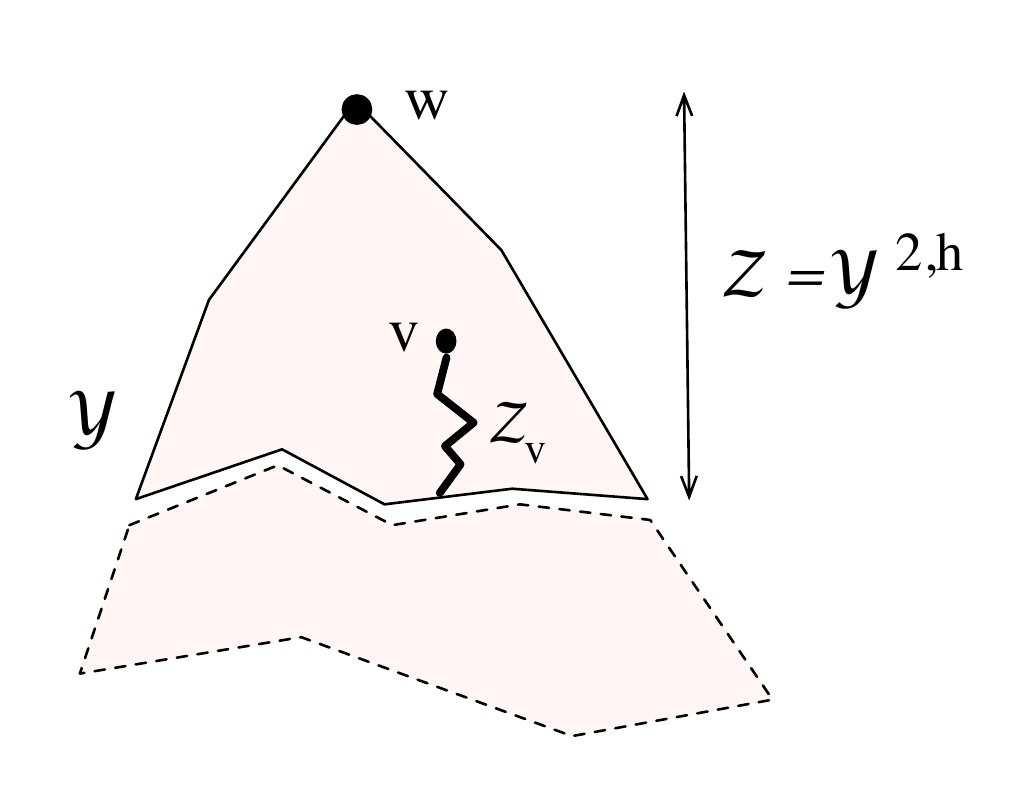}
\caption{On the left, subtree $\calT$, partitioned into layers consisting of nodes with the
same $\kappa$-height. (This figure should not be
interpreted too literally; for example, nodes of $\kappa$-height $3$
may have children of $\kappa$-height $0$ or $1$.)
The figure on the right shows subtree $\calY$ of $\calT$, partitioned
into $\calZ$ and $\calY-\calZ$. Within $\calZ$,
subtree $\calZ_v$ is a path from $v$ to a leaf. 
}
\label{fig: linear tree Y}
\end{center}
\end{figure}

\begin{lemma}\label{lem: lin time in Y}
$\alpha_w \le \bars + O(n)$.
\end{lemma}

\begin{proof}
Let $d = \Depth_2(\calY)$. By Lemma~\ref{lem: gamma-heights},
$d = O(\log|\calY|) = O(\log n)$. 
For $h = 0,...,d$, let $l_h$ be the number of nodes $u\in \calY$
with $\height_2(u,\calY) = h$.
The overall idea of the proof is similar to the analysis of Algorithm~$\algSimAgDTree$.
The difference is that now, since all degrees in $\calY$ are less than $\kappa$,
the number of rounds required to advance through the $h$-th layer of $\calY$, consisting of nodes of
$2$-height equal $h$, can be bounded by $O(m + l_h)$, while before this bound
was $O(n)$. Adding up the bounds for all layers,
all terms $O(l_h)$ will now amortize to $O(n)$,
and the terms $O(m)$ will add up to $O(md)= O(n^{2/3}\log^2n) = O(n)$ as well.
We now fill in the details.

\smallskip
\emparagraph{Claim~A:} 
Let $v$ be a node in $\calY$ with $\height_2(v,\calY) = h$. Then
the activation round of $v$ satisfies
$\alpha_v \le s_h$, where $s_h = \bars + 2n + \sum_{i\le h} l_i + hm$.

\smallskip

First, we observe that Claim~A implies the lemma. This is because for $v=w$ we
get the bound
$\alpha_w \le \bars + 2n + \sum_{i\le d} l_i + dm
			\le \bars + 2n + n + O(\log n)\cdot O(n^{2/3}\log n) = \bars + O(n)$,
as needed.

Thus, to complete the proof, it remains to justify Claim~A. We proceed by
induction on $h$.

Consider first the base case, when $h=0$. We focus on the computation in
the subtree $\calY_v$, which (for $h=0$) is simply a path $v_1,v_2,...,v_q = v$, from a leaf $v_1$
of $\calY$ to $v$. In round $\bars+n$ all the nodes in $\calT-\calY$ must be
already retired. If $v$ is active in round $\bars+n$, we are done, because $\bars+n\le s_0$.
If $v$ is dormant, at least one node in $\calY_v$ must be active (see the invariant
(ii)), so choose $p$
to be the maximum index for which $v_p$ is active. Since
we have no interference from outside $\calY_v$, using a simple inductive
argument, $v$ will be activated in $q-p$ rounds using $\myAll$-transmissions.
Also, $q-p\le l_0$, and therefore 
$\alpha_v \le \bars+n +q-p \le \bars +2n + l_0 = s_0$, which is the bound from 
Claim~A for $h=0$.

In the inductive step, fix some $h>0$, and assume that Claim~A holds for $h-1$.
Denoting $\calZ = \calY^{2,h}$, we consider the computation in 
$\calZ_v$, the subtree of $\calZ$ rooted at $v$. 
(See Figure~\ref{fig: linear tree Y}.)
$\calZ_v$ is a path $v_1,v_2,...,v_q = v$ from a leaf $v_1$ of $\calZ$ to $v$.

The argument is similar to the base case.
There are two twists, however. One, we need to show that $v_1$ will get
activated no later than at time $s_{h-1}+m$; that is, after delay of only $m$,
not $O(n)$. Two,
the children of the nodes on $\calZ_v$ that are not on $\calZ_v$
are not guaranteed to be retired anymore. However, they are semi-retired, 
which is good enough for our purpose.

Consider $v_1$. We want to show first that $v_1$ will get activated
no later than at time $s_{h-1}+m$. All children of $v_1$ can be grouped into
three types. The first type consists of the children of $v_1$ in
$\calT-\calY$. These are activated no later than in round $\bars$, so
they are retired no later than in round $\bars+n$.
All other children of $v_1$ are in $\calY-\calZ$. Among those,
the type-2 children are those that were activated before round $\bars+n$,
and the type-3 children are those that were activated at or after round $\bars+n$.
Clearly, $v_1$ will receive the messages from its children of type~1 and~2,
using $\myRR$-transmissions, no later than in round $\bars+2n$.
The children of $v_1$ of type~3 activate no earlier than in round $\bars+n$.
Also, since they are in $\calY-\calZ$, their $2$-height in $\calY$ is strictly less
than $h$, so they activate no later than in round $s_{h-1}$, by induction. 
(Note that $s_{h-1}\ge \bars+2n$.)
Thus each child $u$ of $v_1$ in $\calY$ of type~3 will complete all its
$\mySel$-transmissions, that include the complete $\kappa$-selector,
between rounds $\bars+n$ and $s_{h-1}+m-1$ (inclusive).
In these rounds all children of $v_1$ that are not in $\calY$ are retired,
so fewer than $\kappa$ children of $v_1$ are active in these rounds.
This implies that the message of $u$ will be received by $v_1$.
Putting it all together, $v_1$ will receive messages from all its children
before round $s_{h-1}+m$, and thus it will be activated no later than in round
$s_{h-1}+m$.

From the paragraph above, we obtain that in round $s_{h-1}+m$
either there is an active node in $\calZ_v$ or all nodes in
$\calZ_v$ are already retired. The remainder of the argument is
similar to the base case.
If $v$ itself is active or retired in round $s_{h-1}+m$ then we are done,
because $s_{h-1}+m\le s_h$.
So suppose that $v$ is still dormant in round $s_{h-1}+m$. 
Choose $p$ to be the largest index for which $v_p$ is active in this round. 
All children of the nodes on $\calZ_v$ that are not on $\calZ_v$ are
either retired or semi-retired. Therefore, since
there is no interference, $v$ will get activated in $q-p$
additional rounds using $\myAll$-transmissions. So 
$\alpha_v \le s_{h-1}+m +q-p 
			\le s_{h-1}+m+ l_h  = s_h$, 
completing the inductive step, the proof of Claim~A, and the lemma.
\end{proof}

From Lemma~\ref{lem: lin time in Y}, all nodes in $\calT$ with $\kappa$-height
equal $0$ will get activated in at most $O(n)$ rounds. For $g=1,2,3$, all
nodes with $\kappa$-height equal $g$ will activate no later than $O(n)$ rounds
after the last node with $\kappa$-height less than $g$ is activated.
This implies that all nodes in $\calT$ will be activated within $O(n)$
rounds. Summarizing, we obtain the main result of this section.


\begin{theorem}\label{thm: DetTree2 time O(n)}
	For any tree with $n$ nodes and any assignment of labels,
	Algorithm~$\algLinAgDTree$ completes information gathering in time $O(n)$.
\end{theorem}


\section{Deterministic Algorithms without Aggregation}
\label{sec: deterministic without aggregation}



In this section we consider deterministic information gathering without
aggregation, where each message can contain at most one rumor, plus additional
$O(\log n)$ bits of information. In this model, we give an algorithm with running
time $O(n\log n)$.

To simplify the description of the algorithm, we
will temporarily assume that we are allowed to receive and
transmit at the same time. Later, we will show how to remove this assumption.


\begin{myalgorithm}{$\algNoAgDTree$}
First, we do some preprocessing, using a modification of
Algorithm~$\algLinAgDTree$ to compute the $2$-height of each node $v$.
In this modified algorithm, the message from each node contains its $2$-height.
When $v$ receives such messages from its children, it can compute its own $2$-height,
which it can then transmit to its parent.

Let $\ell = \ceiling{\log n}$. We divide the computation into $\ell+1$ phases.
Phase $h$, for $h = 0,1,...,\ell$, consists of steps $3nh, 3nh+1,...,3n(h+1)-1$. 
In phase $h$, only the nodes of $2$-height equal $h$ participate in the computation.
Specifically, consider a node $v$ with $\height_2(v) = h$. We have two stages:
\begin{description}
	\item{Stage~$\myAll$:} In each step $t = 3nh,3nh+1,...,3nh+2n-1$, 
		if $v$ contains any rumor $\rho_u$ that it still has not transmitted, 
		$v$ transmits $\rho_u$.
	\item{Stage~$\myRR$:} In each step $t = 3nh+2n+u$, for $u = 0,1,...,n-1$, if
		$v$ has rumor $\rho_u$, then $v$ transmits $\rho_u$.
\end{description}
In any other step, $v$ is in the receiving state.
\end{myalgorithm}


\emparagraph{Analysis.}
By Lemma~\ref{lem: gamma-heights}, the number of phases is $O(\log n)$.
Therefore Algorithm~$\algNoAgDTree$ makes $O(n\log n)$ steps. It remains
to show that $r$ will receive all rumors.

We claim that,
at the beginning of phase $h$, every node $v$ has rumors from all its descendants 
in $\calT - \calT^{2,h}$, namely the descendants
whose $2$-height is strictly smaller than $h$. (In particular, 
if $\height_2(v)<h$ then $v$ has all rumors from $T_v$.) This is trivially true at the
beginning, when $h=0$.

Assume that the claim holds for some $h < \ell$, and consider phase $h$. We want to
show that each node $v$ has rumors from all descendants in $\calT - \calT^{2,h+1}$.
By the inductive assumption, $v$ 
has all rumors from its descendants in $\calT - \calT^{2,h}$.
So if $v$ does not have any descendants of height $h$ (in particular, if
$\height_2(v)\le h-1$) then we are done.

It thus remains to prove that if $v$ has a child $u$ with $\height_2(u) = h$ then
right after phase $h$ all rumors from $\calT_u$ will also be in $v$. (Of course,
this case applies only if $\height_2(v)\ge h$.)

The subtree $\calT^{2,h}_u$, namely the subtree consisting of the descendants of $u$ with 
$2$-height equal $h$,
is a path  $\calP = u_1,u_2,...,u_q = u$, where $u_1$ is a leaf of $\calT^{2,h}$. 
We show that, thanks to pipelining, all rumors that are in $\calP$ when phase $h$ starts
will reach $u$ during Stage~$\myAll$.

In phase $h$, all children of the nodes in $\calP$ that are not on $\calP$ do not transmit.
For any step $3nh+s$, $s =0,1,...,2n-1$, and 
for $i = 1,2,...,q-1$, we define $\phi_{s,i}$ to be the number of rumors in
$u_{i}$ that are still not transmitted, and we let
$\Phi_{s} = \sum_{i=a_s}^{q-1}\max(\phi_{s,i},1)$, where
$a_s$ is the smallest index for which $\phi_{s,a_s}\neq 0$.
We claim that as long as $\Phi_s>0$, its value will decrease in step $s$.
Indeed, for $i<q$,
each node $v_i$ with $\phi_{s,i} > 0$ will transmit a new rumor to
$v_{i+1}$. Since $\phi_{s,i} = 0$ for $i < a_s$, 
node $u_{a_s}$ will not receive any new rumors.
We have $\phi_{a_s} > 0$, by the choice of $a_s$. If
$\phi_{a_s} >1$ then $\max(\phi_{s,a_s},1)$ will decrease by $1$.
If $\phi_{s,a_s} = 1$ then the index $a_s$ itself will increase. In either
case, $u_{a_s}$'s contribution to $\Phi_{s}$ will decrease by $1$.
For $i > a_s$, the term $\max(\phi_{s,i},1)$ cannot increase, because
if $\phi_{s,i} > 0$ then $u_{i}$ transmits a new rumor
to $u_{i+1}$, and if $\phi_{s,i} = 0$ then this term is $1$ anyway.
Therefore, overall, $\Phi_{s}$ will decrease by at least $1$.

Since $\Phi_{s}$ strictly decreases in each step, and its initial value is
at most $q + n \le 2n$, $\Phi_{s}$ will become $0$
in at most $2n$ steps. In other words, in $2n$ steps $u$ will
receive all rumors from $\calP$, and thus all rumors from $\calT_u$.

In Stage~$\myRR$, $u$ will transmit all collected rumors to
$v$, without collisions.
As a result, at the beginning of the next phase $v$ will contain
all rumors from $\calT_u$, completing the proof of the inductive step.

\smallskip
We still need to explain how to modify Algorithm~$\algNoAgDTree$ to
eliminate the assumption that we can transmit and receive at the same
time. This can be accomplished by adding another $3n$ steps to each phase.

We first add a new stage at the beginning of each phase,
consisting of $n$ steps. Note that each node $v$ with 
$\height_2(v) = h$ knows $h$ and also it knows whether it
has a child of $2$-height $h$. If $v$ does not have such a child,
then $v$ is the initial node of a path consisting of nodes of
$2$-height equal $h$. At the very beginning of phase $h$, $v$
sends a message along this path, so that any node on this path
can determine whether it is an even or odd node along this path.

Then we double the number of steps Stage~$\myAll$,
increasing its length from $2n$ to $4n$. In this stage,
among the nodes with $2$-height $h$,
``even'' nodes will transmit in even steps, and ``odd'' nodes
will transmit in odd steps. In this way, each node will
never receive and transmit at the same time.


\begin{theorem}\label{thm: nlog time without aggregation}
For any tree with $n$ nodes and any assignment of labels,
Algorithm~$\algNoAgDTree$ completes information gathering in time $O(n\log n)$.
\end{theorem}


\section{Deterministic Fire-and-Forward Protocols}
\label{sec: det fireandforward protocols}



We now consider a very simple type of protocols that we call \emph{fire-and-forward} protocols.
For convenience, in this model we allow nodes to receive and transmit messages at the same step. 
(Later we will explain how this condition can be removed.) In 
a fire-and-forward protocol, at any time $t$, any node $v$ can either be idle or
make one of two types of transmissions:
\begin{description}
	\item{\emph{Fire}:} $v$ can transmit its own rumor, or
	\item{\emph{Forward}:} $v$ can transmit the rumor received in step $t-1$, if any.
\end{description} 
In Section~\ref{sec: nlogn randomized} we show that there exists a randomized
fire-and-forward protocol that accomplishes information gathering in time $O(n\log n)$.
This raises the question whether this running time can be achieved by a
deterministic fire-and-forward protocol. 
(As before, in the deterministic case we assume that the nodes are labelled $0,1,...,n-1$.)
There is a trivial deterministic fire-and-forward protocol with running time $O(n^2)$: release
all rumors one at a time, spaced at intervals of length $n$. In this section we show that
this can be improved to $O(n^{1.5})$ and that this bound is optimal.

When a node does not forward the rumor received in the previous step, we say that
this rumor is \emph{dropped}. Note that we allow a node to drop a received rumor even
if it does not fire. We will extend the definition of collision to include the
situation when a node attempts to fire right after receiving a rumor (in which case
nothing will be transmitted).

Of course, fire-and-forward protocols use only bounded messages. The key property of
fire-and-forward protocols is that any rumor, once fired,
moves up the tree one hop per step, unless either it collides, or is dropped,
or it reaches the root. 

If rumors fired from two nodes collide at all, they will collide at
their lowest common ancestor. This happens only when the
difference in times between these two firings is equal to the difference
of their depths in the tree. More precisely, let $\calT$ be the tree on input,
denote by $\depth(v)$ the depth of a node $v$ in $\calT$, and suppose that some node
$v$ fires its rumor at time $t$. Then this rumor will reach the root if
no other node $u$ fires at time $t+\depth(v) - \depth(u)$.

The fire-and-forward protocol we develop in this section is \emph{oblivious}, in the sense that
(i) the decision whether to fire or not depends only on the label of the node and the
current time, and (ii) if a node received a rumor in the previous step and it
is supposed to fire at the current step, then no message is transmitted. (This can be
thought of as extending the definition of collisions.)
In fact, it is not hard
to see that any protocol can be turned into an oblivious one without affecting its
asymptotic running time. The idea is that leaves of the tree receive no information at all
during the computation. For any fire-and-forward protocol $\calA$ that runs in time $f(n)$, and for
any tree $\calT$, imagine that we
run this protocol on the tree $\calT'$ obtained by adding a leaf to any
node $v$ and giving it the label of $v$. Label the original nodes with the remaining labels.
This at most
doubles the number of nodes, so $\calA$ will complete in time $O(f(n))$ on $\calT'$.
(We tacitly assume here that $f(cn) = \Theta(n)$ for any constant $c$, which is true for the
time bounds we consider.)
In the execution of $\calA$ on $\calT'$ the leaves receive no information and all rumors
from the leaves will reach the root. This implies that
if we apply $\calA$ on $\calT$ and ignore all information received during the computation,
the rumors will also reach the root. In other words, after this modification, we obtain
an oblivious protocol $\calA'$ with running time $O(f(n))$.


\subsection{An $O(n^{1.5})$ Upper Bound}

We now present our $O(n^{1.5})$-time fire-and-forward protocol. As explained earlier, this protocol 
should specify a set of firing times for each label, so that for any mapping $[n]\to [n]$, 
that maps each label to the depth of the node with this label,
each node will have at least one firing time for which there will not be a collision
along the path to the root. We want each of these firing times to be at most $O(n^{1.5})$.
To this end, we will partition all labels into batches, each of size 
roughly $\sqrt{n}$, and show that for any batch we can define such collision-avoiding
firing times from an interval of length $O(n)$.
Since we have about $\sqrt{n}$ batches, this will give us running time $O(n^{1.5})$.

Our construction of firing times is based on a concept of dispersers, defined below,
which are reminiscent of various \emph{rulers} studied in number theory,
including Sidon sequences~\cite{wikipedia_sidon},
Golomb rulers~\cite{wikipedia_golomb}, or sparse rulers~\cite{wikipedia_sparse}. 
The particular construction we give in the paper is, in a sense, a
multiple set extension of a Sidon-set construction 
by Erd{\"o}s and Tur{\'a}n~\cite{Erdos_Turan_41}.

We now give the details. For $z\in\integers$ and $X\subseteq\integers$, let
$X+z = \braced{x+z\suchthat x\in X}$. Let also $s$ be a positive integer.
A set family $D_1,...,D_m \subseteq [s]$ is called an \emph{$(n,m,s)$-disperser} if 
for each function $\delta: \braced{1,...,m}\to [n]$ and
each $j$ we have $D_j + \delta(j) \not\subseteq \bigcup_{i\neq j} (D_i + \delta(i))$.
The intuition is that $D_j$ represents the set of firing times of node $j$ and
 $\delta(j)$ represents $j$'s depth in the tree.
Then the disperser condition says that some firing in $D_j$ will not collide
with firings of other nodes.


\begin{lemma}\label{lem: disperser}
There exists an $(n,m,s)$-disperser with $m = \Omega(\sqrt{n})$ and $s = O(n)$.
\end{lemma}

\begin{proof}
Let $p$ be the smallest prime such that $p^2 \ge n$. 
For each $a = 1,2,...,p-1$ and $x \in [p]$ define
\begin{equation*}
	d_a(x) = (ax \bmod{p}) + 2p\cdot(ax^2 \bmod{p}).
\end{equation*}
We claim that for any $a\neq b$ and any $t\in\integers$ the equation 
$d_a(x) - d_b(y) = t$ has at most two solutions $(x,y)\in [p]^2$.
For the proof, fix $a,b,t$ and one solution $(x,y)\in [p]^2$. Suppose that
$(u,v)\in [p]^2$ is a different solution. Thus we have
$d_a(x) - d_b(y) = d_a(u) - d_b(v)$. After substituting and 
rearranging, this can be written as
\begin{align*}
(ax\bmod{p}) - (by\bmod{p}) &- (au\bmod{p}) + (bv \bmod{p})
\\
			&= 2p[\,-(ax^2\bmod{p}) + (by^2\bmod{p}) + (au^2\bmod{p}) - (bv^2\bmod{p}) \,].
\end{align*}
The expression on the left-hand side is strictly between $-2p$ and $2p$, so
both sides must be equal $0$. This implies that
\begin{align}
ax-au \,&\equiv\, by-bv \pmod{p}
				\quad\textrm{and}
				\label{eqn: n^1.5 first eq}
		\\
ax^2 - au^2 \,&\equiv\, by^2-bv^2 \pmod{p}.
				\label{eqn: n^1.5 second eq}
\end{align}
From equation~(\ref{eqn: n^1.5 first eq}), the assumption that
$(x,y)\neq (u,v)$ implies that $x\neq u$ and $y\neq v$.
We can then divide the two equations, getting
\begin{equation}
x+u \equiv y+v \pmod{p}.
				\label{eqn: n^1.5 third eq}
\end{equation}
With addition and multiplication modulo $p$, $\integers_p$ is a field. Therefore 
for any $x$ and $y$, and any $a \neq b$, equations (\ref{eqn: n^1.5 first eq}) and (\ref{eqn: n^1.5 third eq})
uniquely determine $u$ and $v$, completing the proof of the claim.

\smallskip

Now, let $m = (p-1)/2$ and $s = 2p^2+p$. 
By Bertrand's postulate we have $\sqrt{n}\le p < 2\sqrt{n}$, which implies
that $m = \Omega(\sqrt{n})$ and $s = O(n)$.
For each $i = 1,2,...,m$, define $D_i = \braced{ d_i(x) \suchthat x\in[p]}$. 
It is sufficient to show that the sets $D_1,D_2,...,D_m$ satisfy the
condition of the $(n,m,s)$-disperser. 

The definition of the sets $D_i$ implies that $D_i\subseteq [s]$ for each $i$.
Fix some some $\delta$ and $j$ from the definition of dispersers. It remains to
verify that $D_j + \delta(j) \not\subseteq \bigcup_{i\neq j} (D_i + \delta(i))$.
For $x\in [p]$ and $i\in\braced{1,2,...,m}$, we say that $i$ \emph{kills} $x$
if $d_j(x) + \delta(j) \in D_i+\delta(i)$.
Our earlier claim implies that any $i\neq j$ kills at most two values in $[p]$.
Thus all indices $i\neq j$ kill at most $2(m-1) = p-3$ integers in
$[p]$, which implies that there is some $x\in[p]$ that is not killed by 
any $i$. For this $x$, we will have
$d_j(x) + \delta(j) \notin \bigcup_{i\neq j} (D_i + \delta(i))$,
completing the proof that $D_1,...,D_m$ is indeed an $(n,m,s)$-disperser.
\end{proof}

We now describe our algorithm.


\begin{myalgorithm}{$\algMlessDTree$}
Let $D_1,D_2,...,D_m$ be the $(n,m,s)$-disperser from Lemma~\ref{lem: disperser}. 
We partition all labels (and thus also the corresponding nodes) arbitrarily
into batches $B_1,B_2,...,B_l$, for $l = \ceiling{n/m}$, with each batch $B_i$ having $m$
nodes (except the last batch, that could be smaller).
Order the nodes in each batch arbitrarily, for example according to increasing labels.

The algorithm has $l$ phases. Each phase $q = 1,2,...,l$ consists of $s'=s+n$ steps in the time interval
$[s'(q-1), s'q-1]$. In phase $q$, the algorithm transmits rumors from batch $B_q$,
by having the $j$-th node in $B_q$ fire at each time $s'(q-1) + \tau$, for $\tau\in D_j$.
Note that in the last $n$ steps of each phase none of the nodes fires.
\end{myalgorithm}


\emparagraph{Analysis.}
We now show that Algorithm~$\algMlessDTree$ correctly performs gathering in any $n$-node
tree in time $O(n^{1.5})$.  Since $m = \Omega(\sqrt{n})$, we have
$l = O(\sqrt{n})$. Also, $s' = O(n)$, so the total run time of the protocol is $O(n^{1.5})$.

It remains to show that during each phase $q$ each
node in $B_q$ will have at least one firing that will send its rumor to the root $r$ without collisions.
Fix some tree $\calT$ and let $\delta(j)\in [n]$ be the depth of the
$j$th node in batch $B_q$.
For any batch $B_q$ and any $v\in B_q$, if $v$ is the
$j$th node in $B_q$ then $v$ will fire at times $s'(q-1) + \tau$, for $\tau\in D_j$.
From the definition of dispersers, there is
$\tau \in D_j$ such that $\tau + \delta(j) - \delta(i) \notin D_i$ for each $i\neq j$.
This means that the firing of $v$ at time $s'(q-1) + \tau$ will not collide with any
firing of other nodes in batch $B_q$. Since the batches are separated by empty intervals
of length $n$, this firing will not collide with any firing in other batches.
So $v$'s rumor will reach $r$. 

Summarizing, we obtain our main result of this section.


\begin{theorem}\label{thm: n^1.5 fire-and-forward upper bound}
There is a fire-and-forward protocol for information gathering in trees
with running time $O(n^{1.5})$.
\end{theorem}

It remains to explain how we can extend this result to the model where
nodes are not allowed to receive and transmit at the same time. We only
give a sketch of the argument. The idea is that if some node $w$ transmits a rumor
$\rho_v$ and receives a rumor $\rho_u$, and if $v$ fired at time $t$, then 
$v$ fired at time $t+\depth(v) - \depth(u)+1$. We can then extend the
definition of collisions to include this situation. Incorporating this into
the construction from Lemma~\ref{lem: disperser}, any index $i$ may now
kill more than two $x$'s, but not more than four. So taking $m = \floor{(p-1)/4}$,
we still will always have an $x$ that is not killed by any $i$.
The details will be provided in the full paper.


\subsection{An $\Omega(n^{1.5})$ Lower Bound}

In this section we show that any fire-and-forward protocol needs time $\Omega(n^{1.5})$ to
deliver all rumors to the root for an arbitrary tree. Fix some fire-and-forward protocol
$\calA$. Without loss of generality, as explained earlier in this section, we can
assume that $\calA$ is oblivious, namely that the set $F_v$ of firing times of each
node $v$ is uniquely determined by its label. 

Let $T$ be the running time of $\calA$.
We first give the proof under the assumption that the total number of firings
of $\calA$, among all nodes in $[n]$, 
is at most $T$. Later we will show how to extend our
argument to protocols with an arbitrary number of firings.

We will show that if $T = o(n^{1.5})$ then
$\calA$ will fail even on a ``caterpillar'' tree,
consisting of a path $P$ of length $n$ with $n$ leaves attached to the
nodes of this path. (For convenience we use $2n$ nodes
instead of $n$, but this does not affect the asymptotic lower bound.)
This path $P$ is fixed, and the label assignment to these nodes is not important
for the proof, but, for the ease of reference,
we will label them $n,n+1,...,2n-1$, in order in which they appear on the path,
with node labeled $2n-1$ being the root. 
The leaves have labels from the set $[n] = \braced{0,1,...,n-1}$. 
To simplify the argument we will identify
the labels with nodes, and we will refer to the node with label $l$ simply
as ``node $l$''. (See Figure~\ref{fig: caterpillar}.)
The objective is to show that there is a way to attach
the nodes from $[n]$ to $P$ to make $\calA$ fail, which means that
there is at least one node $w$ whose all firings will collide with
firings from other nodes.

\begin{figure}[ht]
\begin{center}
\includegraphics[width=4.5in]{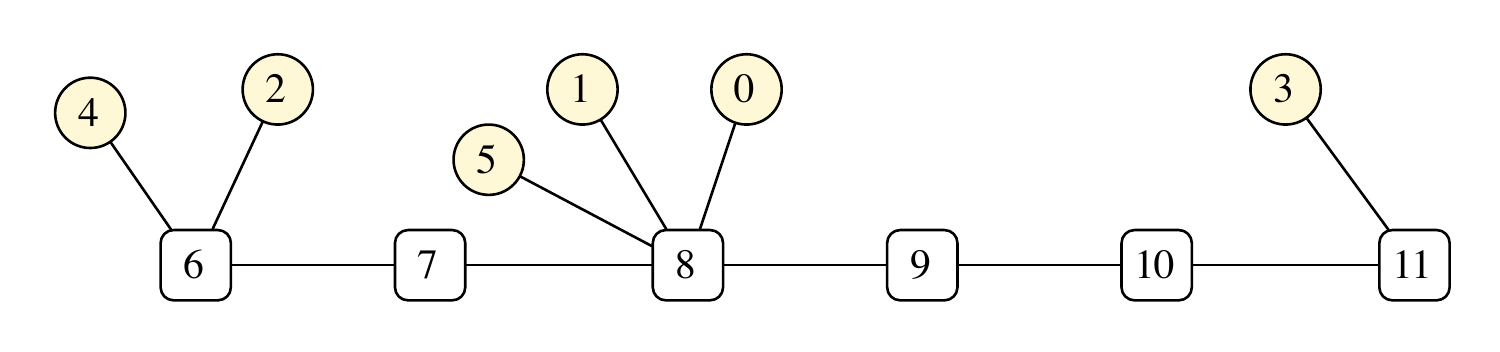}
\caption{A caterpillar graph from the proof, for $n=6$. 
The nodes on the path $P$ are represented by 
rectangles, and the leaves are represented by circles. In this example,
the root is $11$ and $w = 4$.
}
\label{fig: caterpillar}
\end{center}
\end{figure}

Without loss of generality, assume that $T$ is a multiple of $n$.
We let $k = T/n$, and we divide the time range $0,1,...,T-1$ into $k$
bins of size $n$, where the $i$th bin is the interval $[in,(i+1)n-1]$, for $i = 0,1,...,k-1$.
If a node $v\in [n]$ fires at time $t$, we say that a node $u\in [n]-\braced{v}$ 
\emph{covers} this firing if $u$ has a firing at time $t'$ such that $t' > t$ and $t,t'$ are
in the same bin.
For a set $L\subseteq F_v$ of firings of $v$, denote by 
$C(L)$ the set of nodes that cover the firings in $L$.

\paragraph{Claim~A:} For each node $v\in [n]$ there is a set of firings $L_v\subseteq F_v$
such that $|C(L_v)| < |L_v|$.

\smallskip

The proof of Claim~A is by contradiction. Suppose that there exists a node $w\in [n]$
with the property that for each $L\subseteq F_w$ we have $|C(L)|\ge |L|$.
Then Hall's theorem would imply that there is a perfect matching between the firing times 
in $F_w = [n]-\braced{w}$ that cover these firings.
Let the firing times of  $w$ be $F_w = \braced{ t_1, t_2, \dots, t_j}$,
and for each $i = 1,2,...,j$, let $u_i$ be the node matched to $t_i$ in this matching.
By the definitions of bins and covering, each $u_i$
fires at some time $t_i+s_i$, where $0 \leq s_i \leq n-1$.  We can then construct
a caterpillar tree by
attaching $w$ to node $n$ and attaching each $u_i$ to node $n+s_i$ on $P$. 
In this caterpillar tree, the firing of $w$ at each time $t_i$ will collide with
the firing of $u_i$ at time $t_i+s_i$. So the rumor from $w$ will not
reach the root, contradicting the correctness of $\calA$. 
This completes the proof of Claim~A.

\smallskip

For each $v\in [n]$, we now fix a set $L_v$ from Claim~A.
Let $A$ be the set of ordered pairs $(u,v)$ of different nodes $u,v\in [n]$
for which
there is a bin which contains a firing from $L_u$ and a firing from $L_v$,
in this order in time.

We will bound $|A|$ in two different ways. On the one hand, Claim~A implies that 
each $v \in[n]$ appears as the first element in fewer than $|L_v|$ pairs in $A$. 
Adding up over all $v$, and using the assumption that the total number of firings
is at most $T$, we have $|A| <  \sum_v |L_v| \le T$.

On the other hand, we can also establish a lower bound on $|A|$, as follows.
Choose a specific representative firing $t_v$ from each $L_v$. 
For each bin $i$, let $n_i$ be the number of representatives in the $i$th bin.
Any two representatives in bin $i$ contribute one pair to $A$. 
So $|A| \ge q$, for $q= \half\sum_{i=1}^k n_i(n_i-1)$.
Since $\sum_{i=1}^k n_i = n$, if we let all $n_i$ take real values then
the value of $q$ will be minimized when all $n_i$ are equal,
which implies that $q \ge cn^2/k$, for some $c > 0$.
Thus we get that $|A|\ge cn^2/k = cn^3/T$.

Combining the bounds from the last two paragraphs, we obtain that
$T \ge |A| \ge cn^3/T$, which implies that $T = \Omega(n^{1.5})$.
This completes the proof of the lower bound, with the assumption that the total
number of firings does not exceed $T$.

\smallskip

We now consider the general case, without any assumption on the total number
of firings. 
Suppose that $\calA$ is some protocol with running time  $T = o(n^{1.5})$.
Using a probabilistic argument and a reduction to the above special case,
we show that then we can construct a 
caterpillar tree and a node $w$ for which $\calA$ will fail.

Choose some $\tilden$ such that $\tilden = o(n)$ and
$T = o(\tilden^{1.5})$. (For example, we can take $\tilden = n^{1/2}T^{1/3}$.)
Let $\tildeV$ be a random set of nodes, where each node 
is included in $\tildeV$ independently with probability $\tilden/n$.  
Let also $\tildeZ$ be the set of times at which only nodes from $\tildeV$ fire in $\calA$. 
We claim that with probability $1-o(1)$, the following two 
properties will hold:
\begin{description}
\item{(i)} $T \leq |\tildeV|^{1.5}$, and
\item{(ii)} The total number of firings in $\tildeZ$ is at most $T$.
\end{description}
To justify this claim, note the expected cardinality of $\tildeV$ is $\tilden$,
so that probability that (i) is not true is
$\bfP[|\tildeV| < T^{2/3}] = o(1)$, because $T^{2/3} = o(\tilden)$ and the expectation of 
$\tildeV$ is $\tilden$.
To justify (ii), consider some time $t$ where $\calA$ has $j\ge 1$ firings.
Then the probability that $t\in\tildeZ$ is  $(\tilden/n)^j$, so
the contribution of $t$ to the expected number of firings in $\tildeZ$ is $j(\tilden/n)^j = o(1)$.
Therefore the probability that (ii) is violated is $o(1)$.

We can thus conclude that for some choice of $\tildeV$ both properties (i) and (ii)
hold. Let $\barV$ be this choice, $\barn = |\barV|$, and let $\barZ$ be the corresponding
set of times $\tildeZ$.

We convert $\calA$ into a protocol $\calB$ for labels in $\barV$, as follows.
(Note that $\barV$ may not be of the form $[\barn]$, but that does not affect the validity of our argument.)
For each time $t = 0,1,...,T-1$, if $t\in \barZ$ then the firings in $\calB$ are the same as in $\calA$;
if $t\notin\barZ$ then no node in $\barV$ fires.
By our earlier argument, there must be a caterpillar tree $\calT$ with nodes from $\barV$ and a node $w \in \barV$
for which $\calB$ fails. 

Let $\calT'$ be a modified tree obtained from $\calT$ by adding all nodes that are not in $\barV$
as children of the parent of $w$.
Consider now a firing of $w$ in $\calT'$ at time $t$.
If $t \in \barZ$, then this firing  collides with another firing from $\barV$. 
If $t\notin\barZ$, then, by the definition of $\barZ$, there is a node $u$ outside of $\barV$
that fires at the same time (otherwise $t$ would be included in $\barZ$) and is a sibling
of $w$ in $\calT'$. So again, this firing from $w$ will collide with the firing of $u$.
We can thus conclude that the rumor from $w$ will not reach the root, completing the proof of
the lower bound.

We thus obtain our lower bound.


\begin{theorem} \label{thm: fireandforward lower bound}
If $\calA$ is a deterministic fire-and-forward protocol for information gathering
in trees, then the running time of $\calA$ is $\Omega(n^{1.5})$.
\end{theorem}


\section{An $O(n\log n)$-Time Randomized Algorithm}
\label{sec: nlogn randomized}



We now show a randomized algorithm with expected 
running time $O(n\log n)$ that does not use any labels. 
Our algorithm also does not use any aggregation; each message consists only of one rumor
and no additional information.

In the description of the algorithm we assume that the number $n$ of nodes
is known. Using a standard doubling trick, the algorithm can be extended
to one that does not depend on $n$. (This new algorithm will complete
the task in expected time $O(n\log n)$, but it will keep running forever.)

We present the algorithm as a fire-and-forward algorithm (see the previous
section). In particular, we assume for now that at each step a
node can listen and transmit at the same time. We explain later how to
eliminate this feature. Recall that $r$ is the root of the input tree $\calT$.


\begin{myalgorithm}{$\algRandTree$}
Any node $v\neq r$ works like this. At each step $t$, independently of other nodes,
$v$ decides to fire with probability $1/n$.
If the decision is to fire and no rumor arrived at $v$ at step $t-1$, $v$ fires.
If the decision is not to fire and $v$ received some rumor
in the previous step, $v$ forwards this rumor in step $t$.
Otherwise, $v$ is idle.
\end{myalgorithm}

Note that if $v$ decides to fire and it received a rumor in the previous step, then
$v$ will not transmit at all. 


\emparagraph{Analysis.}
We start with the following lemma. 

\begin{lemma}\label{lem: algRandTree prob of z}
At each step $t\ge n$, for each node $z\neq r$,
the probability that $r$ receives rumor $\rho_z$ 
at step $t$ is at least $\frac{1}{n}(1-\frac{1}{n})^{n-1}$.
Furter, for different $t$, the events of $r$ receiving $\rho_z$
are independent.
\end{lemma}

\begin{proof}
To prove the lemma, it helps to view the computation in a slightly different, but
equivalent way. Imagine that we ignore collisions, and we allow each message
to consist of a set of rumors.
If some children of $v$ transmit at a time $t$, then $v$ receives all rumors
in the transmitted messages, and at time $t+1$ it transmits a message
containing all these rumors, possibly adding its own rumor, if $v$
decides to fire at that step. We will refer to these messages as
\emph{virtual messages}. In this interpretation, 
if $r$ receives a virtual message that is a singleton set $\rho_z$, at some time $t$,
then in Algorithm~$\algRandTree$ this rumor $\rho_z$ will be received by $r$
at time $t$. (Note that the converse is not true, because it may happen
that $r$ will receive a virtual message that has other rumors besides $\rho_z$,
but in Algorithm~$\algRandTree$ these other rumors may get stopped earlier
due to collisions, so they will not collide with $\rho_z$.)

Fix  a time $t$ and some $z\in \calT-\braced{r}$. By the above paragraph, it
is sufficient to show that the probability that at time $t$ the
virtual message reaching $r$ is the singleton $\braced{z}$ is equal 
$\frac{1}{n}(1-\frac{1}{n})^{n-1}$.
This event will occur if and only if:
\begin{itemize}
	\item At time $t-\depth(z)$, $z$ decided to fire, and
	\item For each $u\in \calT - \braced{z,r}$,  $u$ did not decide to fire
		at time $t-\depth(u)$.
\end{itemize}
By the algorithm, all these events are independent, so the probability of this combined
event is exactly $\frac{1}{n}(1-\frac{1}{n})^{n-1}$, as needed.
\end{proof}

By Lemma~\ref{lem: algRandTree prob of z}, for any step $t\ge n$,
the probability of any given rumor $\rho_z$ reaching $r$ in step $t$ is
at least as large as the probability of collecting a given
coupon in the coupon collector problem. We thus obtain the
following theorem.


\begin{theorem}\label{thm: rand algorithm}
Algorithm~$\algRandTree$ has expected running time $O(n\log n)$. In fact,
it will complete gathering in time $O(n\log n)$ with probability $1-o(1)$.
\end{theorem}

It remains to argue that we can convert Algorithm~$\algRandTree$ into
the standard model, where receiving and transmitting at the same time is
not allowed. We only give a sketch of the argument here. The algorithm
is essentially the same, with the only difference being that 
if the algorithm decides to fire then it will go into the transmit
state, otherwise it will stay in the receive state. With this change,
some rumors may get rejected, when they are transmitted (without
collision) to a node that happens to be in the transmitting state.
Calculations similar to those above show that this does not change
the asymptotic running time.


\section{An $\Omega(n\log n)$ Lower Bound for Randomized Algorithms}
\label{sec: nlogn randomized lower bound}



In this section, we will show that Algorithm~$\algRandTree$ is within a constant factor of optimal for
label-less algorithms, even if the topology of the tree is known in advance.
Actually we will show something a bit stronger, namely that there is a constant $c$ such
that any label-less algorithm with running time less than $c n \ln n$ will almost surely have some
rumors fail to reach the root on certain trees.

The specific tree we will use here is that of the star graph, consisting of the root with $n$
children that are also the leaves in the tree. (We use $n+1$ nodes instead of $n$, for convenience, but this
does not affect our lower bound.)
In this tree, these leaves are entirely isolated: No leaf receives any information from the root
or from any other leaf.  Thus
at each time step $t$, each leaf $v$ transmits with a probability that
can depend only on $t$ and on the set of previous times at which $v$ attempted to transmit.
Note that, unlike in Algorithm~$\algRandTree$,
the actions of $v$ at different steps may not be independent, which complicates the argument.
Allowing some dependence, in fact, can help reduce the running time, although only by a
constant factor (see Theorem~\ref{thm:nolabelstarbound}).

For the star graph,
we can equivalently think of a label-less algorithm running in time $T$ as
a probability distribution over all subsets of $\braced{0,1,\dots,T-1}$ representing
the sets of transmission times of each node.
Each node $v$ independently picks a subset $S_v$
according to the distribution, and transmits only at the times in $S_v$.  The label-less
requirement is equivalent to the requirement that the $S_v$ are identically distributed.
Node $v$ succeeds in transmitting if there is a time $t$ such that $t \in S_v$,
but $t \notin S_w$ for any $w \neq v$.

The main result of this section is the following lower bound.


\begin{theorem} \label{thm:nolabellowerbound}
If $\calR$ is a randomized protocol for information gathering on trees then
the expected running time of $\calR$ is $\Omega(n\ln n)$.
More specifically, if $n$ is large enough and we run
$\calR$ on the $n$-node star graph for $T\le c n \ln n$ steps,
where $c < \frac{1}{\ln^2{2}}$, then
there will almost surely be at least one rumor that fails to reach the root.
\end{theorem}

As we show, our lower bound is in fact tight, in the sense that the value of the
constant $c$ in the above theorem is best possible for star graphs.

\begin{theorem} \label{thm:nolabelstarbound}
If $T=c n \ln n$, where $c > \frac{1}{\ln^22}$, then there is a protocol which succeeds
on the star graph in time $T$ with probability $1-o(1)$.
\end{theorem}


\subsection{Proof of Theorem~\ref{thm:nolabellowerbound} (Modulo Two Assumptions)}

For now we will assume that our transmission distribution satisfies two additional
simplifying assumptions for every node $v$:
\begin{description}
\item{\textbf{Assumption 1}:} $|S_v| \le  \ln^{10} n$ with probability $1$.
\item{\textbf{Assumption 2}:} $\bfP(t \in S_v) \leq \frac{\ln^4 n}{n}$,
for each time $t$.
\end{description}
Later we will show how to remove these assumptions.

\smallskip

For each time $t$, let $q_t$ equal $n-1$ times the (unconditional) probability that a given
fixed node transmits at time $t$.
Then the conditional probability that a node successfully
transmits its rumor at time $t$, given that it attempted to transmit, is
\begin{equation*}
\textstyle
\left(1-\frac{q_t}{n-1}\right)^{n-1} = e^{-q_t}(1 + O\left(\frac{q_t}{n}\right)).
\end{equation*}
In the rest of the proof we will assume, for simplicity, that $q_t > 0$ for all $t= 0,1,...,T-1$,
since in the steps $t$ when $q_t = 0$ all nodes would be idle, and we can
simply eliminate such steps from the algorithm.

What we would like to do is use this to obtain an estimate on the probability of
successful transmission for a node $v$, given its set $S_v$.
If the transmissions were independent, this would be easy: just take the product
of the failure probabilities at each time in $S_v$ and subtract it from $1$.
As it turns out, this (nearly) gives a lower bound on the failure probability.


\begin{lemma} \label{lemma:monotonicity}
For any distribution that satisfies Assumptions~1 and~2, and any fixed $v$ and $S_v$, we have
\begin{equation*}
\bfP\left(v \textrm{ fails}\, \vline \, S_v\right)
		\geq (1-O(n^{-1+o(1)}))\prod_{t \in S_v} \left(1-e^{-q_t}\right),
\end{equation*}
where ``$v \textrm{ fails}$'' is the event that the rumor from $v$
does not reach the root.
\end{lemma}

The inequality in this lemma only goes in one direction.  As an example, consider the
case where $T=2n$ and each node transmits at times $2t$ and $2t-1$, for a uniformly chosen $t$.
Then if a node's first transmission collides, it makes it significantly more likely
(in fact guaranteed) that the second one does as well.  What the lemma states, in
essence, is that
it cannot make it significantly \emph{less} likely.  We will assume the truth of
Lemma~\ref{lemma:monotonicity} for now and return to its proof later on.

Let $X$ be the random variable representing the number of nodes whose rumors fail
to reach the root. We next show that the expectation of $X$ is large, which is equivalent
to showing that the probability any fixed node $v$ fails is significantly larger than $1/n$.
By Lemma \ref{lemma:monotonicity}, we have
\begin{eqnarray*}
\bfP(v \textrm{ fails }) &=& \bfE_{S_v} \left[\bfP\left(v \textrm{ fails } \vline \, S_v\right)\right]
		\\
&\geq&
\textstyle
(1-O(n^{-1+o(1)})) \bfE_{S_v} \left[\prod_{t \in S_v} \left(1-e^{-q_t}\right)\right]
		\\
&=&
\textstyle
(1-O(n^{-1+o(1)})) \bfE_{S_v} \left[\prod_{t=1}^T \left(1-e^{-q_t} \chi(t \in S_v) \right)\right].
\end{eqnarray*}
Here $\chi(t \in S_v)$ is an indicator function equal to $1$ if $t$ is in $S_v$ and $0$ otherwise.
We now apply a variant of Jensen's inequality which states that $f(\bfE[Y])\ge \bfE[f(Y)]$,
for a random variable $Y$ and a concave function $f(Y)$. Using this inequality, we have
\begin{eqnarray*}
\textstyle
\ln \bfP(v \textrm{ fails}) &\geq&
	\textstyle
	\bfE_{S_v} \left[\ln \left( \prod_{t=1}^T \left(1-e^{-q_t} \chi(t \in S_v)\right)\right)\right]-n^{-1+o(1)}
	\\
&=&
\textstyle
\sum_{t=1}^T \bfE_{S_v} \left[\ln\left(1-e^{-q_t} \chi(t \in S_v)\right)\right]-n^{-1+o(1)}
\\
&=&
\textstyle
\sum_{t=1}^T \frac{q_t}{n-1} \ln\left(1-e^{-q_t}\right)-n^{-1+o(1)}
\\
&\geq&
\textstyle
\frac{T}{n-1} \inf_{x > 0 } \left(x \ln\left(1-e^{-x}\right)\right)-n^{-1+o(1)}
\\
&=& -(c+o(1))\ln^2{2} \ln n.
\end{eqnarray*}
The second inequality holds because $x \ln(1-e^{-x})$, for $x>0$, is minimized
when $x = \ln{2}$.

For any $c< \frac{1}{\ln^2{2}}$, the above bound on $\ln \bfP(v \textrm{ fails})$
implies that any individual node $v$ fails with probability
$n^{-1+\Omega(1)}$, so the expected number of nodes which fail is $\bfE[X] = n^{\Omega(1)}$.

To show concentration around this expectation of $X$, we utilize Talagrand's inequality. Think of $X$
as a function of transmission sets, $X = X(S_1, S_2, \dots, S_n)$.
This function $X$ is Lipschitz, in the sense that changing a single $S_w$ can only change $X$ by at
most $\ln^{10} n$ (by Assumption~1, node $w$ transmits at most $\ln^{10} n$ times, and each transmission
can only interfere with at most one otherwise successful transmission).  Furthermore, $X$ is locally
certifiable in the following sense: If $X \geq x_0$ for some $x_0$, then there is a subset $I$ of at
most  $2 x_0 \ln^{10} n$ nodes such that $X$ remains larger than $x_0$ no matter how we change the
transmission patterns of the nodes outside $I$.  For example, we can construct $I$ as follows.
Start with $I_0$ being a set of $x_0$ nodes that failed to transmit successfully. Then, for each
$v\in I_0$ let $J_v$ be a set of at most $\ln^{10} n$ nodes such that whenever $v$ transmits then
at least one of the nodes in $J_v$ also transmits. Then the set $I = I_0\cup \bigcup_{v\in I_0} J_v$ has the
desired properties.

Let $b$ be the median value of $X$.  It follows from the above two properties of $X$,
together with Talagrand's Inequality \cite{Talagrand_concentration_95}
(see section 7.7 of \cite{Alon_Spencer_Prob_Method_08} for the specific version used here), that
\begin{equation} \label{eqn:talagrand}
\bfP\left(|X-b| \geq \gamma \ln^{15} n \sqrt{b} \right) \leq 4 e^{-\gamma^2/4}.
\end{equation}
To finish, it is enough to show that $b$ is large.  Taking $\gamma =2 \ln n$ in the above inequality,
we have
\begin{equation}
\bfP \left(X \geq b+2 \ln^{16} n \sqrt{b}\right) \leq 4e^{-\ln^2 n}.
		\label{talagrand2}
\end{equation}
Since $X$ is always at most $n$, inequality~(\ref{talagrand2}) implies that
the contribution to $\bfE[X]$ from the values of $X$ that are at least
$b+2 \ln^{16}n \sqrt{b}$ is $n^{o(1)}$.
On the other hand, the contribution to $\bfE[X]$ from the remaining values of $X$ can
be at most $b+2 \ln^{16} n \sqrt{b}$.  It follows that
\begin{equation*}
n^{\Omega(1)} = \bfE(X) \leq b+2 \ln^{16} n \sqrt{b} + n^{o(1)},
\end{equation*}
from which we obtain $b=n^{\Omega(1)}$.
Applying \eqref{eqn:talagrand} with $\gamma =\sqrt{b} \ln^{-15} n$ now gives
\begin{equation*}
\bfP(X=0) \leq \bfP(|X-b| \geq b) \leq 4e^{-b \ln^{-30} n/4}=n^{-\omega(1)},
\end{equation*}
which implies Theorem~\ref{thm:nolabellowerbound}.

\smallskip

It remains to prove Lemma~\ref{lemma:monotonicity} and to remove Assumptions~1 and~2,
that we originally placed on the distribution of transmission times.


\subsection{Proof of Lemma~\ref{lemma:monotonicity}}

The rough idea of the proof of this lemma will be as follows: Suppose we fixed $S_v$ and all of the $q_t$,
and tried to choose our distribution so as to minimize the probability $v$ fails.  Each single transmission
fails with a fixed probability dependent only on $q_t$.  Minimizing the probability they all fail is in
a sense akin to trying to make the individual failures occur at different times.  Intuitively, this should
happen when the other nodes transmit, as much as possible, at different times.

So what we will do is repeatedly transform our distribution to try and separate the transmissions, with
the goal of reaching a distribution where no node transmits at more than one time in $S_v$.
At each step the probability $v$ fails will be non-increasing, and in the end we'll have a distribution
where we can explicitly write down the failure probability.

For each $B \subseteq S_v$, let $y_B$ be the probability that a node $u\neq v$ transmits at the times in $B$,
but does not transmit at any other time in $S_v$.  Suppose that some $U$ with $|U| \geq 2$ has $y_U>0$.
Since each node $u\neq v$ transmits on average fewer than one time in $S_v$ (due to our assumptions bounding the size
 of $S_v$ and the probability a node transmits at a specific time), we must have that $y_{\emptyset} > y_U$.
We now define a new distribution $y'$ as follows:  Fix some $t \in U$, and define
\begin{itemize}
\item{$y_U'=0$}
\item{$y_{U-\{t\}}'=y_{U-\{t\}}+y_U$}
\item{$y_{\{t\}}'=y_{\{t\}}+y_U$}
\item{$y_{\emptyset}'=y_{\emptyset}-y_U$}
\item{for all other $B$, $y_B$ remains unchanged}
\end{itemize}
Effectively what we are doing here is moving mass from $U$ and $\emptyset$ to $U-\{t\}$ and $\{t\}$.
Note that this does not change the values of $q_0,q_1,...,q_{T-1}$.
The intuition above is that this separates the transmission times, so this should reduce the failure
probability of $v$. 
Actually, a stronger monotonicity property is true:


\begin{claim}\label{cla: moving mass}
For any node $w\neq v$, and any transmission pattern of the nodes $u\notin{v,w}$, the probability
that $v$ fails if $w$ transmits according to $y'$ is not larger than the probability that
$v$ fails if $w$ transmits according to $y$.
\end{claim}

The proof of the claim essentially comes from direct verification. Once we have conditioned on
the transmission pattern of the remaining nodes, we will have some set $B'$ corresponding to
the times at which $v$ transmits, but no node apart from $v$ and (possibly) $w$ transmits
(i.e. the times when $v$ would succeed if $w$ did not transmit at all).  Then node $v$ will
fail exactly when $w$ transmits at every time in $B'$.  We now consider a few cases:
\begin{itemize}
\item{If $B'$ is empty, then $v$ has already failed no matter when $w$ does or does not fire.}
\item{If $B'$ contains some time not in $U$, then moving the mass has no impact on the failure
probability (the mass is moved from one set where $v$ succeeds to another where $v$ succeeds).}
\item{If $B' \subseteq U$ and $B'$ contains an element other than $t$, then $v$ fails when $w$
transmits at the times in $U$, but succeeds when $w$ transmits only at time $t$.
So we are moving mass from a pair of elements with at least one failure to one with at most one
failure; therefore we cannot increase the failure probability.}
\item{If $B'$ consists only of the single element $t$, then the operation moves mass from one
success ($U$) and one failure ($\emptyset$) to one success ($\{t\}$) and one failure ($U-\{t\}$),
so the failure probability is unchanged.}
\end{itemize}
This completes the proof of Claim~\ref{cla: moving mass}.

By applying Claim~\ref{cla: moving mass} in succession to each node, we obtain that the failure
probability when all of the nodes transmit according to $y_B'$ is no larger then when they all
transmit according to $y_B$.  By repeatedly applying this operation, we may assume that $y_B > 0$
only when $B$ is either empty or a singleton.

Let $t_1<t_2<\dots<t_k$ be the times in $S_v$.  For each $i= 1,...,k$, let $g_i$ be the
number of nodes other than $v$ that transmit at time $t_i$, and
let $E_i$ be the event that $1\le g_i \le n^{1/6}$.
We can bound the conditional probability that $v$ fails given $S_v$ from below by the probability
that every event $E_i$ occurs, that is
\begin{equation}
	\bfP(v~\textrm{fails}|S_v) \ge \prod_{i=1}^k \bfP\left( E_i \vline \, E_1, \dots, E_{i-1} \right).
		\label{eqn: bound on v failing | Sv}
\end{equation}
We will bound each term in this product directly by separately bounding the conditional probabilities
that $g_i = 0$ and that $g_i > n^{1/6}$.

Conditioned on $E_1, \dots, E_{i-1}$ all occurring, we know that there must be between $n-o(n)$
and $n$ nodes which have not yet transmitted at time $t_i$.
Under our distribution, each of these nodes transmits at time $t_i$ with probability
$\frac{q_{t_i}}{n-1} \left(1-\frac{q_{t_1}+\dots+q_{t_{i-1}}}{n-1}\right)^{-1}$.
This is at least $\frac{q_{t_i}}{n-1}$, and by our two assumptions it is $\frac{q_{t_i}}{n-1}(1+o(1))$
(recall that $q_{t_i}$ is at most $\ln^{4} n$ and $i$ is at most $\ln^{10} n$).
So the probability that $g_i=0$ is at most
\begin{equation*}
\left(1-\frac{q_{t_i}}{n-1}\right)^{n-n^{o(1)}} = e^{-q_{t_i}}\left(1-n^{-1+o(1)}\right),
\end{equation*}
and the probability that $g_i > n^{1/6}$ is at most
\begin{equation*}
\binom{n}{n^{1/6}}\left(\frac{\ln^4 n}{n}\right)^{n^{1/6}}
	\leq \left( \frac{e \ln^4 n}{n^{1/6}}\right)^{n^{1/6}},
\end{equation*}
which by assumption is much smaller than $e^{-q_{t_i}}$.  Combining our two bounds, we have
\begin{equation*}
\bfP\left( E_i \vline \, E_1, \dots, E_{i-1} \right)
		\geq
		\left(1-e^{-q_{t_i}}\right)\left(1-O(n^{-1+o(1)})\right).
\end{equation*}
We now substitute this bound, for each $i$, into inequality~(\ref{eqn: bound on v failing | Sv}).
Once again using Assumption~1, which implies that the product on the right-hand side
of (\ref{eqn: bound on v failing | Sv}) has $n^{o(1)}$ terms,
the lemma is proved.


\subsection{Removing Assumptions~1 and~2}

So far we have shown that any distribution on transmission times satisfying both Assumptions~$1$ and~$2$
will almost surely lead to at least rumor failing to reach the root.  We next consider distributions 
satisfying Assumption~$2$, but not necessarily Assumption~$1$.

For such a distribution, the expected number of times a given node $v$ transmits is
$\bfE[|S_v|] \le T\cdot \ln^4 n/n \le \ln^6 n$ (for $n$ large enough).  
By Markov's inequality, the probability that $v$ transmits at least $\ln^{10} n$ times is $o(1)$.  
We now think of the transmission sets $S_v$ as being generated in a two step exposure process.  
First, we expose a partition of all nodes into two sets, the set $V_1$ of all
nodes $v$ for which $|S_v|\le \ln^{10} n$, and the set $V_2$ of all other nodes.
We then expose the actual set $S_v$ for all nodes in $V_1$.

By Markov's inequality, we have $|V_1| = (1-o(1))n$ almost surely. 
Applying the previous analysis to the conditional distribution on $V_1$
(which by construction now satisfies Assumption~$1$), we almost surely have a node 
$w\in V_1$ which fails 
to transmit successfully (even if the nodes in $V_2$ do not transmit at all).    
This implies that $w$ fails as well if all nodes participate.

It remains to consider distributions which do not satisfy Assumption~$2$.  Suppose there is 
some collection $B_1$ of times at which each node transmits with probability at least $\frac{\ln^4 n}{n}$.  
Let $E_1$ be the event that at least one node transmits successfully in $B_1$, and 
let $E_2$ be the event all nodes transmit successfully in a time not in $B_1$. 
The total success probability is bounded above by $\bfP(E_1)+\bfP(E_2)$.  
By the previous analysis (since Assumption~$2$ is satisfied once the times in $B_1$ are removed), 
we have $\bfP(E_2)=o(1)$.  For the other term, we have
\begin{eqnarray*}
\bfP(E_1) &\leq& \sum_{t \in B_1} \sum_v \bfP(v \textrm{ succeeds at time } t) 
	\\
	&\leq& \sum_{t \in B_1} \sum_v \bfP(\textrm{each node $u\neq v$ does not transmit at time $t$})
	\\
	&\leq& T n \left(1-\frac{\ln^4 n}{n}\right)^{n-1} 
	\\
	&=& (1+o(1)) T n e^{-\ln^4 n} = o(1).
\end{eqnarray*}


\subsection{Proof of Theorem \ref{thm:nolabelstarbound}}

We now show that for $c> \frac{1}{\ln^2{2}}$ there is a protocol which succeeds in time 
$T=c n \ln n$ with probability approaching $1$.  The protocol is straightforward: 
The set $\braced{0,1, \dots, T}$ of possible transmission times is divided into 
$\frac{T \ln 2}{n} = c \ln 2 \ln n$ disjoint intervals of length $\frac{n}{\ln 2}$ each.  
Each node independently and uniformly chooses a single time in each interval at which to transmit.

Fix some arbitrary node $v\neq r$. Under this protocol, the probability that $v$ succeeds
in a given interval is
\begin{equation*}
\left(1-\frac{\ln 2}{n}\right)^{n-1} = \frac{1}{2}-o(1),
\end{equation*}
so the probability that $v$'s rumor fails to reach the root in each interval is
\begin{equation*}
\left(\frac{1}{2}+o(1)\right)^{c \ln 2 \ln n}=n^{-c\ln^2{2}+o(1)}
 = o(n^{-1}),
\end{equation*}
since $c > \frac{1}{\ln^2{2}}$. This holds for each node $v\neq r$, so we can now apply the union bound, 
obtaining that with probability $1-o(1)$ all rumors reach the root.


\bibliographystyle{plain}
\bibliography{gather_trees}

\end{document}